\documentclass[12pt,onecolumn,draftclsnofoot,compsoc]{IEEEtran}
\IEEEoverridecommandlockouts

\usepackage{cite}
\usepackage{balance}
\usepackage{calc}
\usepackage{float}
\usepackage{amsthm}
\usepackage{amsmath,amssymb,amsfonts,mathtools}
\usepackage{algorithmic}
\usepackage{graphicx}
\usepackage{subfigure}
\usepackage{textcomp}
\usepackage{xcolor}
\usepackage{accents}
\usepackage{scalerel}
\usepackage{booktabs}
\usepackage{lettrine}
\usepackage{threeparttable}
\usepackage{bbding}
\usepackage{booktabs, multirow, boldline, hhline, array, nicefrac}
\usepackage{verbatim,placeins}
\usepackage{lmodern}
\usepackage{ifthen}

\newboolean{shortver}
\setboolean{shortver}{false}

\newtheorem{theorem}{Theorem}
\newtheorem{lemma}{Lemma}

\def\BibTeX{{\rm B\kern-.05em{\sc i\kern-.025em b}\kern-.08em T\kern-.1667em\lower.7ex\hbox{E}\kern-.125emX}}

\definecolor{myred}{RGB}{200,0,0}
\definecolor{myblue}{RGB}{0,100,180}
\definecolor{mygreen}{RGB}{0,140,0}
\newcommand{\RED}[1]{\textcolor{myred}{#1}}
\newcommand{\BLUE}[1]{\textcolor{myblue}{#1}}
\newcommand{\GREEN}[1]{\textcolor{mygreen}{#1}}

\newcommand{\Set}[1]{\left\{\, #1 \,\right\}}
\newcommand{\bigSet}[1]{\big\{\, #1 \,\big\}}

\newcommand{\bigO}[0]{\mathcal{O}}
\newcommand{\bigprob}[1]{\Pr\big[#1 \big]}

\newcommand{\msg}[0]{\mathbf{X}}
\newcommand{\indset}[0]{\mathcal{I}}
\newcommand{\que}[1]{\mathbf{Q}^{[#1]}}
\newcommand{\ans}[1]{\mathbf{A}^{[#1]}}
\newcommand{\qspace}[0]{\mathcal{Q}}
\newcommand{\aspace}[0]{\mathcal{A}}
\newcommand{\key}[0]{\kappa}
\newcommand{\perm}[0]{\mathbf{\Pi}}
\newcommand{\dmd}[0]{\mathbf{W}}
\newcommand{\si}[0]{\mathbf{S}}
\newcommand{\ee}[0]{\mathrm{e}}
\newcommand{\ptrn}[0]{\mathbf{F}}
\newcommand{\infserver}[0]{\hat{n}}
\newcommand{\rnd}[1]{\mathbf{#1}}

\newcommand{\downcost}[0]{\mathsf{D}}
\newcommand{\Wdowncost}[0]{\downcost_{\textrm{L-PIR-SI}}^{(W)}}
\newcommand{\WSdowncost}[0]{\downcost_{\textrm{L-PIR-SI}}^{(W, S)}}
\newcommand{\smallmath}[0]{\fontsize{7.5pt}{9pt}\selectfont}

\begin{document}

\title{\Huge On the Leaky Private Information Retrieval \\with Side Information}

\author{
\IEEEauthorblockN{Yingying Huangfu\IEEEauthorrefmark{1} and 
Tian Bai\IEEEauthorrefmark{2}}
\IEEEauthorblockA{\IEEEauthorrefmark{1}Shiled Lab, Huawei Singapore Research Center, Singapore, Singapore\\
\IEEEauthorrefmark{2}Department of Informatics, University of Bergen, Bergen, Norway \\
E-mails:~\IEEEauthorrefmark{1}yingying.huangfu@outlook.com, \IEEEauthorrefmark{2}tian.bai@uib.no
}}


\maketitle

\begin{abstract}
    This paper investigates the problem of Leaky Private Information Retrieval with Side Information~(L-PIR-SI), providing a fundamental characterization of the trade-off among leaky privacy, side information, and download cost.
    We propose a unified probabilistic framework to design L-PIR-SI schemes under $\varepsilon$-differential privacy variants of both $W$-privacy and $(W, S)$-privacy.
    Explicit upper bounds on the download cost are derived, which strictly generalize existing results: our bounds recover the capacity of perfect PIR-SI as $\varepsilon \to 0$, and reduce to the known $\varepsilon$-leaky PIR rate in the absence of side information.
    Furthermore, we conduct a refined analysis of the privacy--utility trade-off at the scaling-law level, demonstrating that the leakage ratio exponent scales as $\bigO(\log \frac{K}{M + 1})$ under leaky $W$-privacy, and as $\bigO(\log K)$ under leaky $(W, S)$-privacy in the minimal non-trivial setting $M = 1$, where $K$ and $M$ denote the number of messages and the side information size, respectively.
\end{abstract}


\section{Introduction}

The private information retrieval~(PIR) problem has attracted significant attention since it was introduced by Chor \textit{et.\ al.}~\cite{CG95} in 1995 due to its critical role in ensuring user privacy during data retrieval from distributed databases.
In the canonical PIR setting~\cite{Sj17}, a user aims to retrieve a specific message from $N$ non-communicating servers, each storing an identical copy of a database containing $K$ messages, without revealing the index of the demand message to any server.
The communication efficiency of PIR schemes is typically measured by the maximum achievable ratio between the demand message size and the number of downloaded symbols.
The maximum of these rates is referred to as the \emph{capacity} and its reciprocal as the \emph{download cost}.
In 2017, the download cost of PIR was fully characterized by Sun and Jafar in~\cite{Sj17} as
\begin{equation*} \label{eq: PIR}
    \downcost_{\textrm{PIR}} = 1 + \frac{1}{N} + \cdots + \frac{1}{N^{K - 1}} = 1 + \frac{1}{N - 1} \cdot \Big(1 - \frac{1}{N^{K - 1}} \Big).
\end{equation*}
More recently, the PIR problem has evolved to incorporate various constraints, including PIR with storage limitations~\cite{ZT20,GZ21,VY23,BZ20,TS18,YS18}, PIR with colluding servers~\cite{SJ18a,BU19a,YL21,HF22,ZTS22}, symmetric PIR~\cite{SunJ19,AU25,ZC24,CLKL22,WBU22}, and multi-message PIR~\cite{BanawanU18a,SiavoshaniSM21,WangBU20}.

One important direction for reducing PIR communication overhead is to leverage side information, where the user possesses prior knowledge of a subset of messages other than the demand one.
This problem, termed PIR with side information (PIR-SI), has been extensively investigated in both single-server~\cite{HeidarzadehGKRS18,HeidarzadehKRS19,LuJ23,Li20} and multi-server configurations~\cite{SiavoshaniSM21,KadheGHRS20,TianS019,LiG20,ChenWJ20a,WeiBU19a}.
The privacy requirement of PIR-SI can be categorized into two types: \textit{$W$-privacy} and \textit{$(W, S)$-privacy}.
The $W$-privacy condition ensures that the index of the demand message remains private from each server, and the $(W, S)$-privacy also protects the index of side information.
The download cost of multi-server PIR-SI under $(W, S)$-privacy was fully characterized in~\cite{ChenWJ20a} as
\begin{equation} \label{eq: PIR-SI WS}
    \downcost_{\textrm{PIR-SI}}^{(W, S)} = 1 + \frac{1}{N - 1} \cdot \Big(1 - \frac{1}{N^{K - M - 1}} \Big),
\end{equation}
with side information size $M$.
However, determining the download cost of PIR-SI under the less restrictive $W$-privacy condition is more challenging.
In 2020, the work in~\cite{KadheGHRS20} established the first upper bound on the download cost as
\ifthenelse{\boolean{shortver}}{
$\downcost_{\textrm{PIR-SI}}^{(W)} \leq 1 + \frac{1}{N - 1} \cdot (1 - \frac{1}{N^{\lceil g \rceil - 1}})$ with $g =\frac{K}{M + 1}$.
}{
\begin{equation} 
    \downcost_{\textrm{PIR-SI}}^{(W)} \leq 1 + \frac{1}{N - 1} \cdot \Big(1 - \frac{1}{N^{\lceil g \rceil - 1}} \Big),
\end{equation}
where $g = \frac{K}{M + 1}$.
}
Later, \cite{LiG20} claimed that this upper bound is tight for certain symmetric and constant-download schemes.
More recently, a probabilistic approach~\cite{WH24} further improved the upper bound to
\begin{equation} \label{eq: PIR-SI W}
    \downcost_{\textrm{PIR-SI}}^{(W)} \leq 1 + \frac{1}{N - 1} \cdot \bigg(1 - \frac{1}{\sum_{k = 0}^{\lceil g \rceil - 1} \binom{g - 1}{k} (N - 1)^{k}} \bigg).
\end{equation}
This result improves upon the download cost upper bound in~\cite{TianS019} whenever $K$ is not divisible by $M + 1$.
To the best of our knowledge, a tight general characterization of the download cost is still unavailable.

Another significant direction for enhancing communication efficiency is to relax the requirement of perfect privacy~\cite{SamyTL19,SamyATL21,HS22,ZhaoHZT25,ChenJJ24,ZhaoHTS25}.
This leads to the study of the leaky PIR~(L-PIR), which allows controlled information leakage about the demand index to reduce download cost.
By optimizing the probability distributions for canonical PIR schemes, several works~\cite{SamyTL19,SamyATL21,ZhaoHTS25} have explored L-PIR in the sense of differential privacy, where the leakage is bounded by a non-negative privacy budget $\varepsilon$~(termed $\varepsilon$-leaky).
The currently best known L-PIR scheme~\cite{ZhaoHTS25} achieves an upper bound on the download cost given by
\begin{equation} \label{eq: L-PIR}
    \downcost_{\textrm{L-PIR}} \leq 1 + \frac{1}{N - 1} \cdot \Big(1 - \frac{1}{((N - 1) \ee^{-\varepsilon} + 1)^{K - 1}} \Big),
\end{equation}
achieving a leakage ratio exponent of $\bigO(\log K)$ with a fixed download cost $\downcost$ and $N$ servers, for the first time in an L-PIR scheme.

\paragraph*{Contribution}

This work investigates the leaky PIR with side information~(L-PIR-SI) problem, aiming to further enhance communication efficiency by integrating the advantages of side information and relaxed privacy constraints.
To this end, we develop a unified probabilistic framework for constructing L-PIR-SI schemes under both the $\varepsilon$-leaky $W$-privacy and $(W, S)$-privacy constraints.
This framework captures the interplay between side-information-aided retrieval and $\varepsilon$-leaky privacy, and yields explicit upper bounds on the download cost.
Under the $\varepsilon$-leaky $W$-privacy condition, the download cost is upper-bounded by
\begin{equation}
    \Wdowncost \leq 1 + \frac{1}{N - 1} \cdot \bigg( 1 - \frac{1}{\sum_{k = 0}^{\lceil g \rceil - 1} \binom{g - 1}{k} (N - 1)^{k} \ee^{-k\varepsilon}} \bigg),
\end{equation}
where $g = \frac{K}{M + 1}$ and $\binom{g - 1}{k}$ is the generalized binomial coefficient.
For the special case $M=1$ under $\varepsilon$-leaky $(W, S)$-privacy, the download cost is upper-bounded by
\begin{equation}
    \WSdowncost  \leq 1 + \frac{1}{N - 1} \cdot \Big( 1 - \frac{1}{((N - 1)\ee^{-\varepsilon} + 1)^{K - 2}} \Big).
\end{equation}

Our schemes not only match known download cost results in the perfect privacy limit ($\varepsilon \to 0$) and the side-information-free case ($M=0$), but also recover the leakage ratio exponents of existing L-PIR schemes at $M = 0$.
In particular, for a fixed download cost $\downcost$ and number of servers $N$, our results yield leakage ratio exponents of $\mathcal{O}(\log \frac{K}{M + 1})$ and $\mathcal{O}(\log K)$ under the $\varepsilon$-leaky $W$-privacy and $(W,S)$-privacy constraints, respectively.
These results provide a seamless generalization of the state-of-the-art performance.

\section{Model and Preliminaries}

We denote by $[a: b]$ the set of integers $\Set{a, a + 1, \ldots, b}$ and use $[a]$ to simply denote the set $[1: a]$.
For a set $\{ X_{1}, X_{2}, \ldots, X_{K} \}$ and an index set $I \subseteq [K]$, we use $X_{I}$ to denote the subset $\{ X_{i} : i \in I \}$.
Particularly, when $I = [a: b]$, we use $X_{a:b}$ to denote the subset $X_{I}$.

\subsection{Problem Formulation}

Consider a network composed of a user and $N \geq 2$ non-communicating database servers, denoted by $[N]$.
Each of which has a copy of $K$ independent messages, denoted by $X_{1:K}$.
The user wishes to retrieve one of these $K$ messages, and without loss of generality, we assume that $K \geq 2$.
Each message contains $L$ sub-packets distributed uniformly on a finite field $\mathbb{F}_{\mathsf{q}}$.
Denote $H(\cdot)$ as the entropy function under the logarithm of base $\mathsf{q}$, we define
\begin{equation}
    L \coloneq H(\msg_{1}) = H(\msg_{2}) = \cdots = H(\msg_{K}),
\end{equation}
called the length of each message.

Suppose the user is interested in retrieving a message $X_{W}$ with $W \in [K]$, given prior knowledge of $M$ messages $X_{S}$, where $S \subseteq [K] \setminus \{W\}$ and $|S| = M$.
Besides, we denote $U \coloneq [K] \setminus (S \cup \{K\})$.
We refer to $W$ as the \textit{demand index}, $S$ as the \textit{side information index} set, and $U$ as the \textit{unknown index} set.
To retrieve the demand message $X_{W}$ given the side information $X_{S}$, the user generates a query $\que{W, S}_{n}$ from the space $\qspace$ and sends it to each server $n \in [N]$.
The query $\que{W, S}_{n}$ is a (deterministic or stochastic) function of $W$, $S$, and $X_{S}$.
This means that the queries are independent of the other messages with indices in $[K] \setminus S$, i.e., for any $n \in [N]$, $W \in [K]$, and $S \subseteq [K] \setminus \{W\}$,
\begin{equation} \label{cond: query vs data}
    I\big(\mathbf{X}_{1:K}; \que{\dmd,\si}_{n} \mid W, S, X_{S}\big) = 0.
\end{equation}

Upon receiving the query, the server $n$ responds to the user with an answer $\ans{W, S}_{n}$ taken from the space $\aspace$.
The answer $\ans{W, S}_{n}$ is a deterministic function of $\que{W, S}_{n}$ and $X_{1:K}$, i.e., for any query $q \in \qspace$, it holds
\begin{equation} \label{cond: answer vs query}
  H\big(\ans{\dmd,\si}_{n} \mid \que{\dmd,\si}_{n} = q, X_{1:K}\big) = 0.
\end{equation}
We refer to the set of queries and the answers as the L-PIR-SI scheme.

The \textit{correctness condition} requires that the user can successfully decode the message $X_{W}$ from the answers $\ans{W, S}_{1:N}$ and the side information $X_{S}$.
That is, for any answer $a \in \aspace$, $W \in [K]$ and $S \subseteq [K] \setminus \{W\}$, it holds
\begin{equation} \label{cond: correctness} 
    H\big(\rnd{X}_{W} \mid \ans{W, S}_{1:N} = a, X_{S}\big) = 0.
\end{equation}

\subsection{Leaky Privacy}

For each database $n \in [N]$, our focus privacy condition in this paper is a leaky setting of the $W$-privacy and $(W, S)$-privacy, which is a relaxation of the perfect privacy condition.
Under the perfect privacy requirement~\cite{Sj17,ChenWJ20a,WH24,TianS019}, any query $\que{W, S}_{n}$ does not reveal any information about the demand index $W$ to server $n$.
In a leaky privacy setting, we allow a certain amount of information about the demand index $W$ to be leaked to each server.

Following the previous works on L-PIR~\cite{SamyTL19,SamyATL21,ZhaoHTS25}, we consider the leakage measured via differential privacy and adopt the following leaky conditions for $W$-privacy and $(W, S)$-privacy of the L-PIR-SI schemes, respectively. 

Formally speaking, an L-PIR-SI scheme satisfies \emph{$\varepsilon$-leaky $W$-privacy} if for all $n \in [N]$, $W, W' \in [K]$, query $q \in \qspace$, and answer $a \in \aspace$, it holds
\begin{equation} \label{cond: W-privacy}
    \frac{\bigprob{\que{\dmd,\si}_{n} = q, \ans{\dmd, \si}_{n} = a\mid W}}{\bigprob{\que{\dmd,\si}_{n} = q, \ans{\dmd, \si}_{n} = a \mid W'}} \leq \ee^{\varepsilon},
\end{equation}
where $\varepsilon$ is called the leakage ratio exponent.
Note that this condition may also allow for a leakage of a certain amount of information regarding the side information index set $S$. 
When $\varepsilon = 0$, it degenerates to perfect $W$-privacy.

In addition, an L-PIR-SI scheme satisfies \emph{$\varepsilon$-leaky $(W, S)$-privacy} if for each server $n \in [N]$, $W, W' \in [K]$, $S \subseteq [K] \setminus \{W\}$, $S' \subseteq [K] \setminus \{W'\}$, query $q \in \qspace$, and answer $a \in \aspace$, it holds
\begin{equation} \label{cond: WS-privacy}
    \frac{\bigprob{\que{\dmd, \si}_{n} = q, \ans{\dmd, \si}_{n} = a \mid W, S}}{\bigprob{\que{\dmd, \si}_{n} = q, \ans{\dmd, \si}_{n} = a \mid W', S' } } \leq \ee^{\varepsilon},
\end{equation}
where $\varepsilon$ is called the leakage ratio exponent.
Note that this condition not only protects the demand index $W$ but also the side information index set $S$ from being leaked to each server.
When $\varepsilon = 0$, it degenerates to perfect $(W, S)$-privacy.

In this work, we measure the download efficiency of an L-PIR-SI scheme by the \textit{normalized average download cost},
defined as the total number of downloaded symbols from all servers per message symbol retrieved, i.e.,
\begin{equation}
    \downcost \coloneq \max_{W \in [K], S \subseteq [K] \setminus \{W\}} \frac{1}{L} \cdot \sum\nolimits_{n=1}^{N} H\big(\ans{W, S}_{n}\big).
\end{equation}
The \emph{download cost} of the L-PIR-SI problem under $W$-privacy or $(W, S)$-privacy, respectively denoted by $\Wdowncost$ and $\WSdowncost$, is then defined as the minimum of this quantity over all feasible L-PIR-SI schemes for the given $N$, $K$, and $M$.
Note that the download cost is the inverse of the capacities defined in~\cite{ZhaoHTS25}.

\section{Main Results}


\begin{theorem}\label{thm: W-privacy}
For $\varepsilon$-leaky $W$-privacy, the L-PIR-SI download cost is upper-bounded by
    \begin{equation}\label{eq: W-privacy}
        \Wdowncost \leq 1 + \frac{1}{N - 1} \cdot \bigg( 1 - \frac{1}{\sum_{k = 0}^{\lceil g \rceil - 1}\binom{g - 1}{k}(N - 1)^{k} \ee^{-k\varepsilon}} \bigg),
    \end{equation}
    where $g =  \frac{K}{M + 1}$, and $\binom{g - 1}{k} = \frac{(g - 1)(g - 2) \cdots (g - k)}{k!}$ is the generalized binomial coefficient.
\end{theorem}

Remark that our result in Theorem~\ref{thm: W-privacy} generalizes several existing results in the PIR literature, bridging PIR-SI and L-PIR. 
Specifically, when $\varepsilon = 0$, the download cost matches Eq.~\eqref{eq: PIR-SI W}, corresponding to PIR-SI under perfect $W$-privacy~\cite{WH24}.
Additionally, when $M = 0$, which implies that $g = K$ is an integer, the summation $\sum_{k = 0}^{g - 1} \binom{g - 1}{k} (N - 1)^{k} \ee^{-k\varepsilon}$ can be simplified via the binomial theorem to $((N - 1)\ee^{-\varepsilon} + 1)^{K - 1}$, yielding a download cost that matches the L-PIR result in~\cite{ZhaoHTS25}~(cf.  Eq.~\eqref{eq: L-PIR}).

By rearranging the download cost expression shown in Eq.~\eqref{eq: W-privacy}, the scaling behavior of the privacy leakage can be stated in the following theorem.
\begin{theorem} \label{thm: leakage W-privacy}
For $\varepsilon$-leaky $W$-privacy, the L-PIR-SI leakage ratio exponent can be bounded as
    \ifthenelse{\boolean{shortver}}{
    $\varepsilon = \bigO(\log \frac{K}{M + 1})$.
    }{
    \begin{equation}
        \varepsilon = \bigO\Big(\log \frac{K}{M + 1}\Big).
    \end{equation}
    }
\end{theorem}
Note that the parameter $g = \frac{K}{M + 1}$ can be interpreted as an \emph{equivalent} number of messages, since the download cost expression resembles that of a PIR problem with $g$ messages without side information.
Theorem~\ref{thm: leakage W-privacy} indicates that for a fixed download cost $\downcost$ and number of servers $N$, the leakage ratio exponent scales as $\bigO(\log g)$.
This result generalizes the scaling behavior observed in leaky-PIR without side information~\cite{ZhaoHTS25}.

\begin{theorem}\label{thm: WS-privacy}
For $\varepsilon$-leaky $(W, S)$-privacy with $M = 1$, the L-PIR-SI download cost is upper-bounded by
    \begin{equation}\label{eq: WS-privacy}
        \WSdowncost \leq 1 + \frac{1}{N - 1} \cdot \Big( 1 - \frac{1}{((N - 1)\ee^{-\varepsilon} + 1)^{K - 2}} \Big).
    \end{equation}
\end{theorem}

Theorem~\ref{thm: WS-privacy} implies that when $\varepsilon = 0$, the download cost matches the PIR-SI result shown in Eq.~\eqref{eq: PIR-SI WS} for $M = 1$, corresponding to PIR-SI under the perfect $(W, S)$-privacy condition characterized in~\cite{ChenWJ20a}.
Similarly, the download cost expression Eq.~\eqref{eq: WS-privacy} can be rearranged to characterize the scaling behavior of the privacy leakage as follows.
\begin{theorem} \label{thm: leakage WS-privacy}
For $\varepsilon$-leaky $(W, S)$-privacy with $M = 1$, the L-PIR-SI leakage ratio exponent can be bounded as
    \ifthenelse{\boolean{shortver}}{
    $\varepsilon = \bigO(\log K)$.
    }{
    \begin{equation}
        \varepsilon = \bigO(\log K).
    \end{equation}
    }
\end{theorem}
Theorem~\ref{thm: leakage WS-privacy} also generalizes the scaling behavior observed in leaky-PIR without side information~\cite{ZhaoHTS25}.

\ifthenelse{\boolean{shortver}}{
The Proofs of Theorems~\ref{thm: leakage W-privacy} and \ref{thm: leakage WS-privacy} can be found in the full version~\cite{HuangB26}.
}{
The proofs of Theorems~\ref{thm: leakage W-privacy} and \ref{thm: leakage WS-privacy} are separately provided in the Appendix~\ref{app: leakage WS-privacy} and~\ref{app: leakage W-privacy}.
}

\section{The Achievability Proof}

\subsection{General Construction}

This subsection introduces a general \emph{probabilistic} framework for constructing L-PIR-SI schemes.

In the framework, each message $X_{i}$ consists of $L = N - 1$ sub-packets, i.e.,
\ifthenelse{\boolean{shortver}}{$X_{i} = [X_{i}[1], X_{i}[2], \ldots, X_{i}[L]] \in \mathbb{F}_{\mathsf{q}}^{L}$,}
{
\begin{equation}
    X_{i} = [X_{i}[1], X_{i}[2], \ldots, X_{i}[L]] \in \mathbb{F}_{\mathsf{q}}^{L},
\end{equation}
}
and $X_{1:K}$ can be partitioned into three parts $X_{W}$, $X_{U}$, $X_{S}$ as defined earlier.
Each server is queried to return the sum of several specified sub-packets from \emph{distinct} messages, subject to the following rules based on $(X_{W}, X_{U}, X_{S})$:
\begin{itemize}
    \item Exactly one server $\infserver$, called the \emph{inference server}, is randomly selected to exclude any sub-packet of $X_{W}$, while every other server includes a distinct randomly selected sub-packet of $X_{W}$.
    \item All servers include the same randomly generated sub-packets of $X_{U}$.
    \item Each server independently includes additional randomly selected sub-packets of $X_{S}$.
\end{itemize}

Now, we formally introduce the framework.
For notational convenience, we append a \emph{dummy symbol} to each message by defining $X_{i}[0] = 0$, and thus extend the sub-packet index set to $\indset = [0: N - 1]$.
The randomness of our scheme is expressed as a random pattern $\key = \big(\perm, \ptrn_{U}, \ptrn^{(0)}_{S}, \ptrn^{(1)}_{S}\big)$ such that
\begin{itemize}
    \item $\perm$ is a random bijective mappings $[N] \to \indset$, and the unique inference server is determined as $\infserver = \perm^{-1}(0)$.
    \item $\ptrn_{U} = \big[\ptrn_{i}\big]_{i \in U} \in \indset^{|U|}$ is a random vector such that the queries of all servers contain the same indices of the sub-packets from $X_{i}$ for each unknown index $i \in U$.
    \item $\ptrn^{(0)}_{S} = \big[\ptrn^{(0)}_{i}\big]_{i \in S} \in \indset^{M}$ (resp., $\ptrn^{(1)}_{S} = \big[\ptrn^{(1)}_{i}\big]_{i \in S} \in \indset^{M}$) is a random vector such that the query of the inference server (resp., non-inference server) independently contains the indices of the sub-packets from $X_{i}$ for each side information index $i \in S$.
\end{itemize}

For a server $n$ and a realization $\big(\pi, f_{U}, f^{(0)}_{S}, f^{(1)}_{S}\big)$ of $\key$, we let $f^{(0)}_{U} = f^{(1)}_{U} = f_{U}$ and $f^{(0)}_{W} = f^{(1)}_{W} = \pi(n)$ for notational convenience.
Then, in our framework, the query of server $n$ can be expressed as a vector $Q^{[W, S]}_{n} \coloneq q_{n} = \big[q_{n}[1], q_{n}[2], \ldots, q_{n}[K]\big]$, which is defined as
\ifthenelse{\boolean{shortver}}{
\begin{equation}
    \begin{cases}
        \big[f^{(0)}_{1}, \ldots, f^{(0)}_{K}\big], &\text{if~} n = \infserver~\text{(inference server)}; \\
        \big[f^{(1)}_{1}, \ldots, f^{(1)}_{K}\big], &\text{if~} n \neq \infserver~\text{(non-inference server)}.
    \end{cases}
\end{equation}
}{
\begin{equation}
    \begin{cases}
        \big[f^{(0)}_{1}, f^{(0)}_{2}, \ldots, f^{(0)}_{K}\big], &\text{if~} n = \infserver~\text{(inference server)}; \\
        \big[f^{(1)}_{1}, f^{(0)}_{2}, \ldots, f^{(1)}_{K}\big], &\text{if~} n \neq \infserver~\text{(non-inference server)}.
    \end{cases}
\end{equation}
}

Here, the $i$th component $q_{n}[i]$ specifies the index of the requested sub-packet from message $X_{i}$. 
Particularly, $q_{n}[i] = 0$ means that $X_{i}$ does not contribute to the answer from server $n$.

Upon receiving the query $Q^{[W, S]}_{n}$, server $n$ computes the answer $a_{n}$ as a linear combination of the requested sub-packets, which is expressed as
\ifthenelse{\boolean{shortver}}{
\begin{equation}
    a_{n} =\sum_{i = 1}^{K} X_{i}\big[q_{n}[i]\big] =
    \begin{cases}
        \sum_{i=1}^{K} X_{i}\big[f^{(0)}_{i}\big], & \text{if~} n = \infserver; \\
        \sum_{i=1}^{K} X_{i}\big[f^{(1)}_{i}\big], & \text{if~} n \neq \infserver.
    \end{cases}
\end{equation}
}{
\begin{equation}
    a_{n} =\sum\nolimits_{i = 1}^{K} X_{i}\big[q_{n}[i]\big] =
    \begin{cases}
        \sum\nolimits_{i=1}^{K} X_{i}\big[f^{(0)}_{i}\big], & \text{if~} n = \infserver~\text{(inference server)}; \\
        \sum\nolimits_{i=1}^{K} X_{i}\big[f^{(1)}_{i}\big], & \text{if~} n \neq \infserver~\text{(non-inference server)}.
    \end{cases}
\end{equation}
}

\subsection{Correctness and Download Cost}\label{section: correctness and download cost}

The correctness of the proposed framework holds regardless of the specific design of the random pattern $\key$.
Since the user knows all the side information $X_{S}$, their contribution can be removed from the answers. 
Moreover, all servers include exactly the same sub-packets from $X_{U}$, so their combined contribution is identical across servers.
Therefore, the user can eliminate this interference by subtracting the answer of the inference server from that of any non-inference server.
Specifically, for any realization $(\pi, f_{U}, f^{(0)}_{S}, f^{(1)}_{S})$ of $\key$, each demand sub-packet $X_{W}[i]$ with $i \in [L]$ can be recovered by the answers of the server $n = \pi^{-1}(i)$ and inference server $\infserver = \pi^{-1}(0)$ by
\begin{equation}\label{eq: decode}
    \begin{aligned}
        X_{W}[i]
        ={}&\big( a_{n} - \sum\nolimits_{i \in S} X_{i}\big[q_{n}[i]\big] \Big) - \Big(a_{\infserver} - \sum\nolimits_{i \in S} X_{i}\big[q_{\infserver}[i]\big] \Big) \\
        ={}&\Big( a_{n} - \sum\nolimits_{i \in S} X_{i}\big[f^{(0)}_{i}\big] \Big) - \Big(a_{\infserver} - \sum\nolimits_{i \in S} X_{i}\big[f^{(1)}_{i}\big] \Big). 
    \end{aligned}
\end{equation}
By iterating over all $i \in [L]$, the user finally reconstructs the entire demand message $X_{W}$.

Next, we analyze the download cost of the proposed framework.
In the scheme, the answer $\ans{\dmd, \si}_{n}$ for each server $n \in [N]$ is empty or the linear combination of sub-packets from different messages.
Additionally, the answer is empty if and only if the server $n$ is the inference server and the realizations of $\mathbf{F}_{U}$ and $\mathbf{F}_{S}^{(0)}$ are all-zero vectors.
Denote $P_{0}$ as the probability that the answer $\ans{\dmd, \si}_{\infserver}$ is empty for the inference server $\infserver$.
Then the download cost $\downcost_{\key}$ of our framework can be expressed as
\begin{equation} \label{eq: download cost}
    \downcost_{\key} 
    = \frac{P_{0}(N - 1) + (1 - P_{0})N}{L} 
    = 1 + \frac{1}{N - 1} \cdot (1 - P_{0}).
\end{equation}

Now, we are ready to present specific probability distributions of the random pattern $\key$ to achieve $\varepsilon$-leaky $W$-privacy and $(W, S)$-privacy, respectively.

\subsection{The Leaky $W$-Privacy Scheme}

\ifthenelse{\boolean{shortver}}{This subsection presents a probability scheme for L-PIR-SI achieving $\varepsilon$-leaky $W$-privacy.}
{This subsection presents a probability scheme for L-PIR-SI achieving $\varepsilon$-leaky $W$-privacy.
The main idea of the design is to control the Hamming weight of each query such that it is either a multiple of $M + 1$ or equal to $K$.
Moreover, a query sent to a non-inference server includes exactly $M$ non-zero side information indices, whereas a query sent to an inference server is designed to include as few non-zero side information indices as possible (to satisfy the above constraint for the overall Hamming weight).
Given any pair $(W, S)$, the conditional probability of generating a query with a fixed Hamming weight at the inference server is at most $\ee^{\varepsilon}$ that at a non-inference server.}

The distribution of the random pattern $\key_{2} = \big(\perm, \ptrn_{U}, \ptrn^{(0)}_{S}, \ptrn^{(1)}_{S}\big)$ is defined as follows.

First, $\perm$ is a uniformly random bijection from $[N]$ to $\mathcal{I}$.

Second, we define a specific $\ptrn_{U}$.
Let $g = \frac{K}{M + 1} $ and $\ell_{k} = \min\{k(M + 1), K - M - 1\}$.
The Hamming weight of $\ptrn_{U}$ is chosen from
\begin{equation*}
    \bigSet{ \ell_{k}}_{k = 0}^{\lceil g \rceil - 1} = \bigSet{ 0, M + 1, \ldots, (\lceil g \rceil - 2)(M + 1), K - M - 1 }.
\end{equation*}
The random variable $k$ is drawn according to the distribution $\Set{P_{k}}_{k  = 0}^{\lceil g \rceil - 1}$ such that
\begin{equation} \label{W-privacy: P}
    P_{k} \coloneq \frac{\binom{g - 1}{k} (N - 1)^{k} \ee^{-k\varepsilon}}{\sum\nolimits_{k = 0}^{\lceil g \rceil - 1} \binom{g - 1}{k} (N - 1)^{k} \ee^{-k\varepsilon}}.
\end{equation}
One can easily verify that $\sum_{k = 0}^{\lceil g \rceil - 1} P_{k} = 1$.
Moreover, conditioned on having Hamming weight $\ell_{k}$, the vector $\ptrn_{U}$ is uniformly distributed over all vectors of weight $\ell_{k}$.

Third, $\ptrn^{(1)}_{S}$ is a uniform distribution over $[N - 1]^{M}$, and $\ptrn^{(0)}_{S}$ is a uniform distribution over $[N - 1]^{k(M + 1) - \ell_{k}}$.
Actually, $\ptrn^{(0)}_{S}$ must be a single-point distribution unless $k = \lceil g \rceil- 1$ and $K$ cannot be divided by $M + 1$.

We give an example with $N = 3$, $K = 3$, $M = 1$, and $L = 2$ to illustrate the proposed scheme.
Then, we have $X_{i} = [X_{i}[1], X_{i}[2]]$ for $i \in [3]$.
Assume that the user wants to retrieve message $X_{1}$ with side information $X_{2}$~(i.e., $W = 1$, $S = \Set{2}$, $U = \Set{3}$, and $g = \frac{3}{2}$). 
Note that $\perm$ takes $3! = 6$ possible permutations with equal probability.
We only present the realizations of the random pattern $\key_{1} = \big(\pi, f_{U}, f^{(0)}_{S}, f^{(1)}_{S}\big)$ with $\pi = (0, 1, 2)$.
Under this realization, the inference server is Server $1$.
The possible realizations of the random pattern $\key_{1}$ are listed in Table~\ref{res: w-privacy1}.
\ifthenelse{\boolean{shortver}}{}
{
    Further examples are provided in Appendix~\ref{app: example}.
}

\begin{table*}[!tbp]
    \caption{Realizations of $\key_{1} = \big(\pi, f_{U}, f^{(0)}_{S}, f^{(1)}_{S}\big)$ with $\pi = (0, 1, 2)$ and $(N, K, M) = (3, 3, 1)$ for $\varepsilon$-leaky $W$-privacy. \\
            Here, $P_{0} = \frac{1}{\ee^{-\varepsilon} + 1}$ and $P_{1} = \frac{\ee^{-\varepsilon}}{\ee^{-\varepsilon} + 1}$.}
    \label{res: w-privacy1}
    \centering
    \setlength{\tabcolsep}{3pt}
    \begin{tabular}{l@{\hspace{3pt}}l@{\hspace{1pt}}l c@{\hspace{2pt}}c c@{\hspace{2pt}}c c@{\hspace{2pt}}c c}
        \toprule
        \multirow{2}{*}{\smallmath $f^{\phantom{()}}_{U}$}       & 
        \multirow{2}{*}{\smallmath $f^{(0)}_{S}$} & 
        \multirow{2}{*}{\smallmath $f^{(1)}_{S}$} & \multicolumn{2}{c}{Server $1$}                                                   & \multicolumn{2}{c}{Server $2$}                                                                              & \multicolumn{2}{c}{Server $3$}                                                                              & \multirow{2}{*}{Prob.}               \\
        \cline{4-9}
                  &           &            & \multicolumn{1}{c}{$q_{1}$}           & $a_{1}$                                         & \multicolumn{1}{c}{$q_{2}$}           & $a_{2}$                                                             & \multicolumn{1}{c}{$q_{3}$}           & $a_{3}$                                                             &                                      \\
        \midrule
        \,\RED{0} & ~\BLUE{0} & ~\GREEN{1} & \multicolumn{1}{c}{0\BLUE{0}\RED{0}}  & $\varnothing$                                   & \multicolumn{1}{c}{1\GREEN{1}\RED{0}} & \smallmath $X_{1}[1]\!+\!X_{2}[\GREEN{1}]\phantom{{}\!+\!X_{3}[0]}$ & \multicolumn{1}{c}{2\GREEN{1}\RED{0}} & \smallmath $X_{1}[2]\!+\!X_{2}[\GREEN{1}]\phantom{{}\!+\!X_{3}[0]}$ & \multirow{2}{*}{$\frac{1}{6} P_{0}$} \\ 
        \,\RED{0} & ~\BLUE{0} & ~\GREEN{2} & \multicolumn{1}{c}{0\BLUE{0}\RED{0}}  & $\varnothing$                                   & \multicolumn{1}{c}{1\GREEN{2}\RED{0}} & \smallmath $X_{1}[1]\!+\!X_{2}[\GREEN{2}]\phantom{{}\!+\!X_{3}[0]}$ & \multicolumn{1}{c}{2\GREEN{2}\RED{0}} & \smallmath $X_{1}[2]\!+\!X_{2}[\GREEN{2}]\phantom{{}\!+\!X_{3}[0]}$ &                                      \\ 
        \hline
        \,\RED{1} & ~\BLUE{1} & ~\GREEN{1} & \multicolumn{1}{c}{0\BLUE{1}\RED{1}}  & \smallmath $X_{2}[\BLUE{1}]\!+\!X_{3}[\RED{1}]$ & \multicolumn{1}{c}{1\GREEN{1}\RED{1}} & \smallmath $X_{1}[1]\!+\!X_{2}[\GREEN{1}]\!+\!X_{3}[\RED{1}]$       & \multicolumn{1}{c}{2\GREEN{1}\RED{1}} & \smallmath $X_{1}[2]\!+\!X_{2}[\GREEN{1}]\!+\!X_{3}[\RED{1}]$       & \multirow{8}{*}{$\frac{1}{6} P_{1}$} \\ 
        \,\RED{1} & ~\BLUE{1} & ~\GREEN{2} & \multicolumn{1}{c}{0\BLUE{1}\RED{1}}  & \smallmath $X_{2}[\BLUE{1}]\!+\!X_{3}[\RED{1}]$ & \multicolumn{1}{c}{1\GREEN{2}\RED{1}} & \smallmath $X_{1}[1]\!+\!X_{2}[\GREEN{2}]\!+\!X_{3}[\RED{1}]$       & \multicolumn{1}{c}{2\GREEN{2}\RED{1}} & \smallmath $X_{1}[2]\!+\!X_{2}[\GREEN{2}]\!+\!X_{3}[\RED{1}]$       &                                      \\ 
        \,\RED{1} & ~\BLUE{2} & ~\GREEN{1} & \multicolumn{1}{c}{0\BLUE{2}\RED{1}}  & \smallmath $X_{2}[\BLUE{2}]\!+\!X_{3}[\RED{1}]$ & \multicolumn{1}{c}{1\GREEN{1}\RED{1}} & \smallmath $X_{1}[1]\!+\!X_{2}[\GREEN{1}]\!+\!X_{3}[\RED{1}]$       & \multicolumn{1}{c}{2\GREEN{1}\RED{0}} & \smallmath $X_{1}[2]\!+\!X_{2}[\GREEN{1}]\!+\!X_{3}[\RED{1}]$       &                                      \\   
        \,\RED{1} & ~\BLUE{2} & ~\GREEN{2} & \multicolumn{1}{c}{0\BLUE{2}\RED{1}}  & \smallmath $X_{2}[\BLUE{2}]\!+\!X_{3}[\RED{1}]$ & \multicolumn{1}{c}{1\GREEN{2}\RED{1}} & \smallmath $X_{1}[1]\!+\!X_{2}[\GREEN{2}]\!+\!X_{3}[\RED{1}]$       & \multicolumn{1}{c}{2\GREEN{2}\RED{1}} & \smallmath $X_{1}[2]\!+\!X_{2}[\GREEN{2}]\!+\!X_{3}[\RED{1}]$       &                                      \\ 
        \,\RED{2} & ~\BLUE{1} & ~\GREEN{1} & \multicolumn{1}{c}{0\BLUE{1}\RED{2}}  & \smallmath $X_{2}[\BLUE{1}]\!+\!X_{3}[\RED{2}]$ & \multicolumn{1}{c}{1\GREEN{1}\RED{2}} & \smallmath $X_{1}[1]\!+\!X_{2}[\GREEN{1}]\!+\!X_{3}[\RED{2}]$       & \multicolumn{1}{c}{2\GREEN{1}\RED{2}} & \smallmath $X_{1}[2]\!+\!X_{2}[\GREEN{1}]\!+\!X_{3}[\RED{2}]$       &                                      \\ 
        \,\RED{2} & ~\BLUE{1} & ~\GREEN{2} & \multicolumn{1}{c}{0\BLUE{1}\RED{2}}  & \smallmath $X_{2}[\BLUE{1}]\!+\!X_{3}[\RED{2}]$ & \multicolumn{1}{c}{1\GREEN{2}\RED{2}} & \smallmath $X_{1}[1]\!+\!X_{2}[\GREEN{2}]\!+\!X_{3}[\RED{2}]$       & \multicolumn{1}{c}{2\GREEN{1}\RED{2}} & \smallmath $X_{1}[2]\!+\!X_{2}[\GREEN{2}]\!+\!X_{3}[\RED{2}]$       &                                      \\           
        \,\RED{2} & ~\BLUE{2} & ~\GREEN{1} & \multicolumn{1}{c}{0\BLUE{2}\RED{2}}  & \smallmath $X_{2}[\BLUE{2}]\!+\!X_{3}[\RED{2}]$ & \multicolumn{1}{c}{1\GREEN{1}\RED{2}} & \smallmath $X_{1}[1]\!+\!X_{2}[\GREEN{1}]\!+\!X_{3}[\RED{2}]$       & \multicolumn{1}{c}{2\GREEN{1}\RED{2}} & \smallmath $X_{1}[2]\!+\!X_{2}[\GREEN{1}]\!+\!X_{3}[\RED{2}]$       &                                      \\                          
        \,\RED{2} & ~\BLUE{2} & ~\GREEN{2} & \multicolumn{1}{c}{0\BLUE{2}\RED{2}}  & \smallmath $X_{2}[\BLUE{2}]\!+\!X_{3}[\RED{2}]$ & \multicolumn{1}{c}{1\GREEN{2}\RED{2}} & \smallmath $X_{1}[1]\!+\!X_{2}[\GREEN{2}]\!+\!X_{3}[\RED{2}]$       & \multicolumn{1}{c}{2\GREEN{1}\RED{2}} & \smallmath $X_{1}[2]\!+\!X_{2}[\GREEN{2}]\!+\!X_{3}[\RED{2}]$       &                                      \\
        \bottomrule
    \end{tabular}
\end{table*}

\begin{table*}[!tbp]
    \caption{Realizations of $\key_{2} = \big(\pi, f_{U}, f^{(0)}_{S}, f^{(1)}_{S}\big)$ with $\pi = (0, 1, 2)$ and $(N, K, M) = (3, 3, 1)$ for $\varepsilon$-leaky $(W,S)$-privacy. \\
            Here, $P_{0} = \frac{1}{2\ee^{-\varepsilon} + 1}$ and $P_{1} = \frac{2\ee^{-\varepsilon}}{2\ee^{-\varepsilon} + 1}$.}
    \label{res: ws-privacy1}
    \centering
    \setlength{\tabcolsep}{3pt}
    \begin{tabular}{l@{\hspace{3pt}}l@{\hspace{1pt}}l c@{\hspace{2pt}}c c@{\hspace{2pt}}c c@{\hspace{2pt}}c c}
        \toprule
        \multirow{2}{*}{\smallmath $f^{\phantom{()}}_{U}$} & 
        \multirow{2}{*}{\smallmath $f^{(0)}_{S}$} & 
        \multirow{2}{*}{\smallmath $f^{(1)}_{S}$} & \multicolumn{2}{c}{Server $1$}                                                  & \multicolumn{2}{c}{Server $2$}                                                                              & \multicolumn{2}{c}{Server $3$}                                                                              & \multirow{2}{*}{Prob.}                \\
        \cline{4-9}
                  &            &           & \multicolumn{1}{c}{$q_{1}$}          & $a_{1}$                                         & \multicolumn{1}{c}{$q_{2}$}           & $a_{2}$                                                             & \multicolumn{1}{c}{$q_{3}$}           & $a_{3}$                                                             &                                       \\
        \midrule
        \,\RED{0} & ~\BLUE{0} & ~\GREEN{1} & \multicolumn{1}{c}{0\BLUE{0}\RED{0}} & $\varnothing$                                   & \multicolumn{1}{c}{1\GREEN{1}\RED{0}} & \smallmath $X_{1}[1]\!+\!X_{2}[\GREEN{1}]\phantom{{}\!+\!X_{3}[0]}$ & \multicolumn{1}{c}{2\GREEN{1}\RED{0}} & \smallmath $X_{1}[2]\!+\!X_{2}[\GREEN{1}]\phantom{{}\!+\!X_{3}[0]}$ & \multirow{2}{*}{$\frac{1}{6} P_{0}$} \\
        \,\RED{0} & ~\BLUE{0} & ~\GREEN{2} & \multicolumn{1}{c}{0\BLUE{0}\RED{0}} & $\varnothing$                                   & \multicolumn{1}{c}{1\GREEN{2}\RED{0}} & \smallmath $X_{1}[1]\!+\!X_{2}[\GREEN{2}]\phantom{{}\!+\!X_{3}[0]}$ & \multicolumn{1}{c}{2\GREEN{2}\RED{0}} & \smallmath $X_{1}[2]\!+\!X_{2}[\GREEN{2}]\phantom{{}\!+\!X_{3}[0]}$ &                                       \\
        \hline
        \,\RED{1} & ~\BLUE{1} & ~\GREEN{0} & \multicolumn{1}{c}{0\BLUE{1}\RED{1}} & \smallmath $X_{2}[\BLUE{1}]\!+\!X_{3}[\RED{1}]$ & \multicolumn{1}{c}{1\GREEN{0}\RED{1}} & \smallmath $X_{1}[1]\phantom{{}\!+\!X_{2}[0]}\!+\!X_{3}[\RED{1}]$   & \multicolumn{1}{c}{2\GREEN{0}\RED{1}} & \smallmath $X_{1}[2]\phantom{{}\!+\!X_{2}[0]}\!+\!X_{3}[\RED{1}]$   & \multirow{12}{*}{$\frac{1}{6} P_{1}$}\\
        \,\RED{1} & ~\BLUE{1} & ~\GREEN{1} & \multicolumn{1}{c}{0\BLUE{1}\RED{1}} & \smallmath $X_{2}[\BLUE{1}]\!+\!X_{3}[\RED{1}]$ & \multicolumn{1}{c}{1\GREEN{1}\RED{1}} & \smallmath $X_{1}[1]\!+\!X_{2}[\GREEN{1}]\!+\!X_{3}[\RED{1}]$       & \multicolumn{1}{c}{2\GREEN{1}\RED{1}} & \smallmath $X_{1}[2]\!+\!X_{2}[\GREEN{1}]\!+\!X_{3}[\RED{1}]$       &                                       \\
        \,\RED{1} & ~\BLUE{1} & ~\GREEN{2} & \multicolumn{1}{c}{0\BLUE{1}\RED{1}} & \smallmath $X_{2}[\BLUE{1}]\!+\!X_{3}[\RED{1}]$ & \multicolumn{1}{c}{1\GREEN{2}\RED{1}} & \smallmath $X_{1}[1]\!+\!X_{2}[\GREEN{2}]\!+\!X_{3}[\RED{1}]$       & \multicolumn{1}{c}{2\GREEN{2}\RED{1}} & \smallmath $X_{1}[2]\!+\!X_{2}[\GREEN{2}]\!+\!X_{3}[\RED{1}]$       &                                       \\
        \,\RED{1} & ~\BLUE{2} & ~\GREEN{0} & \multicolumn{1}{c}{0\BLUE{2}\RED{1}} & \smallmath $X_{2}[\BLUE{2}]\!+\!X_{3}[\RED{1}]$ & \multicolumn{1}{c}{1\GREEN{0}\RED{1}} & \smallmath $X_{1}[1]\phantom{{}\!+\!X_{2}[0]}\!+\!X_{3}[\RED{1}]$   & \multicolumn{1}{c}{2\GREEN{0}\RED{1}} & \smallmath $X_{1}[2]\phantom{{}\!+\!X_{2}[0]}\!+\!X_{3}[\RED{1}]$   &                                       \\
        \,\RED{1} & ~\BLUE{2} & ~\GREEN{1} & \multicolumn{1}{c}{0\BLUE{2}\RED{1}} & \smallmath $X_{2}[\BLUE{2}]\!+\!X_{3}[\RED{1}]$ & \multicolumn{1}{c}{1\GREEN{1}\RED{1}} & \smallmath $X_{1}[1]\!+\!X_{2}[\GREEN{1}]\!+\!X_{3}[\RED{1}]$       & \multicolumn{1}{c}{2\GREEN{1}\RED{1}} & \smallmath $X_{1}[1]\!+\!X_{2}[\GREEN{1}]\!+\!X_{3}[\RED{1}]$       &                                       \\
        \,\RED{1} & ~\BLUE{2} & ~\GREEN{2} & \multicolumn{1}{c}{0\BLUE{2}\RED{2}} & \smallmath $X_{2}[\BLUE{2}]\!+\!X_{3}[\RED{1}]$ & \multicolumn{1}{c}{1\GREEN{2}\RED{2}} & \smallmath $X_{1}[1]\!+\!X_{2}[\GREEN{2}]\!+\!X_{3}[\RED{1}]$       & \multicolumn{1}{c}{2\GREEN{0}\RED{2}} & \smallmath $X_{1}[2]\!+\!X_{2}[\GREEN{2}]\!+\!X_{3}[\RED{1}]$       &                                       \\
        \,\RED{2} & ~\BLUE{1} & ~\GREEN{0} & \multicolumn{1}{c}{0\BLUE{1}\RED{2}} & \smallmath $X_{2}[\BLUE{1}]\!+\!X_{3}[\RED{2}]$ & \multicolumn{1}{c}{1\GREEN{0}\RED{2}} & \smallmath $X_{1}[1]\phantom{{}\!+\!X_{2}[0]}\!+\!X_{3}[\RED{2}]$   & \multicolumn{1}{c}{2\GREEN{0}\RED{2}} & \smallmath $X_{1}[2]\phantom{{}\!+\!X_{2}[0]}\!+\!X_{3}[\RED{2}]$   &                                       \\
        \,\RED{2} & ~\BLUE{1} & ~\GREEN{1} & \multicolumn{1}{c}{0\BLUE{1}\RED{2}} & \smallmath $X_{2}[\BLUE{1}]\!+\!X_{3}[\RED{2}]$ & \multicolumn{1}{c}{1\GREEN{1}\RED{2}} & \smallmath $X_{1}[1]\!+\!X_{2}[\GREEN{1}]\!+\!X_{3}[\RED{2}]$       & \multicolumn{1}{c}{2\GREEN{1}\RED{2}} & \smallmath $X_{1}[1]\!+\!X_{2}[\GREEN{1}]\!+\!X_{3}[\RED{2}]$       &                                       \\
        \,\RED{2} & ~\BLUE{1} & ~\GREEN{2} & \multicolumn{1}{c}{0\BLUE{1}\RED{2}} & \smallmath $X_{2}[\BLUE{1}]\!+\!X_{3}[\RED{2}]$ & \multicolumn{1}{c}{1\GREEN{2}\RED{2}} & \smallmath $X_{1}[1]\!+\!X_{2}[\GREEN{2}]\!+\!X_{3}[\RED{2}]$       & \multicolumn{1}{c}{2\GREEN{2}\RED{2}} & \smallmath $X_{1}[2]\!+\!X_{2}[\GREEN{2}]\!+\!X_{3}[\RED{2}]$       &                                       \\            
        \,\RED{2} & ~\BLUE{2} & ~\GREEN{0} & \multicolumn{1}{c}{0\BLUE{1}\RED{2}} & \smallmath $X_{2}[\BLUE{2}]\!+\!X_{3}[\RED{2}]$ & \multicolumn{1}{c}{1\GREEN{0}\RED{2}} & \smallmath $X_{1}[1]\phantom{{}\!+\!X_{2}[0]}\!+\!X_{3}[\RED{2}]$   & \multicolumn{1}{c}{2\GREEN{0}\RED{2}} & \smallmath $X_{1}[2]\phantom{{}\!+\!X_{2}[0]}\!+\!X_{3}[\RED{2}]$   &                                       \\
        \,\RED{2} & ~\BLUE{2} & ~\GREEN{1} & \multicolumn{1}{c}{0\BLUE{2}\RED{2}} & \smallmath $X_{2}[\BLUE{2}]\!+\!X_{3}[\RED{2}]$ & \multicolumn{1}{c}{1\GREEN{1}\RED{2}} & \smallmath $X_{1}[1]\!+\!X_{2}[\GREEN{1}]\!+\!X_{3}[\RED{2}]$       & \multicolumn{1}{c}{2\GREEN{1}\RED{2}} & \smallmath $X_{1}[2]\!+\!X_{2}[\GREEN{1}]\!+\!X_{3}[\RED{2}]$       &                                       \\                            
        \,\RED{2} & ~\BLUE{2} & ~\GREEN{2} & \multicolumn{1}{c}{0\BLUE{2}\RED{2}} & \smallmath $X_{2}[\BLUE{2}]\!+\!X_{3}[\RED{2}]$ & \multicolumn{1}{c}{1\GREEN{2}\RED{2}} & \smallmath $X_{1}[1]\!+\!X_{2}[\GREEN{2}]\!+\!X_{3}[\RED{2}]$       & \multicolumn{1}{c}{2\GREEN{2}\RED{2}} & \smallmath $X_{1}[2]\!+\!X_{2}[\GREEN{2}]\!+\!X_{3}[\RED{2}]$       &                                       \\
        \bottomrule
    \end{tabular}
\end{table*}

Notice that the proposed scheme is a specific configuration of the general framework introduced in the previous subsection. 
Essentially, for $k \in [0: \lceil g \rceil- 2]$, the Hamming weight of the $\ptrn^{(0)}_{S}$ and $\ptrn^{(1)}_{S}$ is $0$ and $M$, respectively.
Thus, Eq.~\eqref{eq: decode} can be simplified as
\begin{equation}
    X_{W}[i] = a_{n} - \Big(a_{\infserver} - \sum\nolimits_{i \in S} X_{i}\big[f^{(0)}_{i}\big] \Big).
\end{equation}

\begin{lemma} \label{lem: W-privacy}
    The proposed L-PIR-SI scheme under $\key_{1}$ is $\varepsilon$-leaky $W$-privacy.
\end{lemma}
\begin{proof}
    \ifthenelse{\boolean{shortver}}
    {See the full version~\cite{HuangB26} for the complete proof.
    }{
    See Appendix~\ref{app: W-privacy}.}
\end{proof}

Furthermore, we have $P_{0}^{-1} = \sum_{k = 0}^{\lceil g \rceil - 1} \binom{g - 1}{k} (N - 1)^{k} \ee^{- k \varepsilon}$ for $k = 0$ in Eq.~\eqref{W-privacy: P}.
Substituting $P_{0}$ into Eq.~\eqref{eq: download cost}, the download cost of this proposed scheme is
\begin{equation}
    \downcost_{\key_{1}} = 1 + \frac{1}{N - 1} \cdot \bigg( 1 - \frac{1}{\sum_{k = 0}^{\lceil g \rceil - 1} \binom{g - 1}{k} (N - 1)^{k} \ee^{-k\varepsilon}} \bigg).
\end{equation}

Finally, combining the correctness analysis in subsection~\ref{section: correctness and download cost}, we have established Theorem~\ref{thm: W-privacy}.

\subsection{The Leaky $(W, S)$-Privacy Scheme}

\ifthenelse{\boolean{shortver}}{This subsection presents a probability scheme for L-PIR-SI achieving $\varepsilon$-leaky $(W, S)$-privacy for a special case of $M = 1$.}
{This subsection presents a probability scheme for L-PIR-SI achieving $\varepsilon$-leaky $(W, S)$-privacy for a special case of $M = 1$.
The main idea of the design is to control the Hamming weight of each query such that it is either $0$ or no less than $M + 1$.
Given any pair $(W, S)$, the conditional probability of generating a query with a fixed Hamming weight at the inference server is at most $\ee^{\varepsilon}$ that at a non-inference server.
}

The distribution of the random pattern $\key_{2} = \big(\perm, \ptrn_{U}, \ptrn^{(0)}_{S}, \ptrn^{(1)}_{S}\big)$ is defined as follows.

First, $\perm$ is a uniformly random bijection from $[N]$ to $\mathcal{I}$.

Second, the Hamming weight $\ell$ of $\ptrn_{U}$ is chosen from the set $[0: K - 2]$.
The random variable $\ell$ is drawn according to the distribution $\Set{P_{\ell}}_{\ell = 0}^{K - 2}$ with
\begin{equation}\label{eq:WS-probability}
    P_{\ell} 
    = \frac{\binom{K - 2}{\ell} \cdot (N - 1)^{\ell}\ee^{-\ell \varepsilon}}{((N - 1)\ee^{-\varepsilon} + 1)^{K - 2}}.
\end{equation}
One can easily verify that $\sum_{\ell = 0}^{K - 2} P_{\ell} = 1$.
Moreover, conditioned on having Hamming weight $\ell$, the vector $\ptrn_{U}$ is uniformly distributed over all vectors of weight $\ell$.

Third, for $j \in [0:1]$, the Hamming weight $s_{j}$ of $\ptrn^{(j)}_{S}$ is either $0$ or $1$.
So that the Hamming weight of the query can be expressed as $\ell + s_{j} + j$.
The distribution of $s_{j}$ is determined by the conditional probability
\ifthenelse{\boolean{shortver}}
{$\Pr[s_{j} \mid \ell, j] = \frac{r^{\ell + s_{j} + j} + (-1)^{\ell + s_{j} + j} \cdot r}{(r + 1) \cdot r^{\ell + j}}$, where $r = (N - 1)\ee^{-\varepsilon}$.}
{
\begin{equation}
    \Pr[s_{j} \mid \ell, j] 
    = \frac{r^{\ell + s_{j} + j} + (-1)^{\ell + s_{j} + j} \cdot r}{(r + 1) \cdot r^{\ell + j}},
\end{equation}
where $r = (N - 1)\ee^{-\varepsilon}$.
}
It is easy to see that the conditional probability is non-negative as $r \geq 1$.
Additionally, one can easily verify that $\Pr[s_{j} = 0 \mid \ell, j] + \Pr[s_{j} = 1 \mid \ell, j] = 1$.
Finally, conditioned on having $s_{j}$ ($j \in [0:1]$), the vector $\ptrn^{(j)}_{S}$ is uniformly distributed over all vectors of weight $s_{j}$.

To illustrate the proposed scheme, we provide an example with the same parameter settings in Table~\ref{res: w-privacy1}.
We also focus on the specific realization of the random pattern $\key_2 = (\pi, f_{U}, f^{(0)}_{S}, f^{(1)}_{S})$ with $\pi = (0, 1, 2)$, where Server $1$ acts as the inference server.
All possible realizations of the random pattern $\key_2$ are detailed in Table~\ref{res: ws-privacy1}.
\ifthenelse{\boolean{shortver}}{}
{
    Further examples are provided in Appendix~\ref{app: example}.
}

\begin{lemma} \label{lem: WS-privacy}
     For $M = 1$, the proposed L-PIR-SI scheme under $\key_{2}$ is $\varepsilon$-leaky $(W, S)$-privacy.
\end{lemma}
\begin{proof}
    \ifthenelse{\boolean{shortver}}
    {See the full version~\cite{HuangB26} for the complete proof.
    }{
    See Appendix~\ref{app: WS-privacy}.}
\end{proof}

We have $P_{0}^{-1} = ((N - 1)\ee^{-\varepsilon} + 1)^{K - 2}$ for $\ell = 0$ in Eq.~\eqref{eq:WS-probability}.
Substituting $P_{0}$ into Eq.~\eqref{eq: download cost}, the download cost of this proposed scheme is
\begin{equation}
    \downcost_{\key_{2}} = 1 + \frac{1}{N - 1} \cdot \Big( 1 - \frac{1}{((N - 1)\ee^{-\varepsilon} + 1)^{K - 2}} \Big).
\end{equation}
Then, combining with the correctness analysis in subsection~\ref{section: correctness and download cost}, we have established Theorem~\ref{thm: WS-privacy}.

\section{Conclusion}

This paper proposed a unified probabilistic framework for $\varepsilon$-leaky PIR-SI under $W$-privacy and $(W, S)$-privacy. 
Our download cost bounds not only recover the best-known limits in the perfect privacy and the side-information-free case, but also capture the scaling behavior of the leakage ratio exponent. 
In the absence of a general converse for PIR-SI, these bounds provide meaningful performance guarantees for leaky retrieval.  
For $(W, S)$-privacy, our framework yields an elegant bound for $M = 1$, while deriving the matching download costs for general $M \ge 2$ remains an interesting open problem.

\balance

\bibliographystyle{IEEEtran}
\bibliography{PIR.bib}

\ifthenelse{\boolean{shortver}}
{}
{\newpage
\appendices

\section{The L-PIR-SI Schemes with $N=3$, $K=4$ and $M=1$}\label{app: example}

\begin{table*}[!tbp]
    \caption{Realizations of $\key_{1} = \big(\pi, f_{U}, f^{(0)}_{S}, f^{(1)}_{S}\big)$ with $\pi = (0, 1, 2)$ and $(N, K, M) = (3, 4, 1)$ for $\varepsilon$-leaky $W$-privacy. \\
            Here, $P_{0} = \frac{1}{2\ee^{-\varepsilon} + 1}$ and $P_{1} = \frac{2\ee^{-\varepsilon}}{2\ee^{-\varepsilon} + 1}$.}
    \label{res2: w-privacy}
    \centering    
    \setlength{\tabcolsep}{3pt}
    \begin{tabular}{l@{\hspace{3pt}}l@{\hspace{1pt}}l c@{\hspace{2pt}}c c@{\hspace{2pt}}c c@{\hspace{2pt}}c c}
        \toprule
        \multirow{2}{*}{\smallmath $f^{\phantom{()}}_{U}$} & 
        \multirow{2}{*}{\smallmath $f^{(0)}_{S}$} & 
        \multirow{2}{*}{\smallmath $f^{(1)}_{S}$} & \multicolumn{2}{c}{Server $1$}                                                 & \multicolumn{2}{c}{Server $2$}                                                                                                & \multicolumn{2}{c}{Server $3$}                                                                                            & \multirow{2}{*}{Prob.}                \\
        \cline{4-9}
                 &           &            & \multicolumn{1}{c}{$q_{1}$}           & $a_{1}$                                        & \multicolumn{1}{c}{$q_{2}$}            & $a_{2}$                                                                              & \multicolumn{1}{c}{$q_{3}$}            & $a_{3}$                                                                          &                                       \\
        \midrule
        \RED{00} & ~\BLUE{0} & ~\GREEN{1} & \multicolumn{1}{c}{0\BLUE{0}\RED{00}} & $\varnothing$                                  & \multicolumn{1}{c}{1\GREEN{1}\RED{00}} & \smallmath $X_{1}[1]\!+\!X_{2}[\GREEN{1}]\phantom{{}\!+\!X_{3}[0]\!+\!X_{4}[0]}$     & \multicolumn{1}{c}{2\GREEN{1}\RED{00}} & \smallmath $X_{1}[2]\!+\!X_{2}[\GREEN{1}]\phantom{{}\!+\!X_{3}[0]\!+\!X_{4}[0]}$ & \multirow{2}{*}{$\frac{1}{24} P_{0}$} \\ 
        \RED{00} & ~\BLUE{0} & ~\GREEN{2} & \multicolumn{1}{c}{0\BLUE{0}\RED{00}} & $\varnothing$                                  & \multicolumn{1}{c}{1\GREEN{2}\RED{00}} & \smallmath $X_{1}[1]\!+\!X_{2}[\GREEN{2}]\phantom{{}\!+\!X_{3}[0]\!+\!X_{4}[0]}$     & \multicolumn{1}{c}{2\GREEN{2}\RED{00}} & \smallmath $X_{1}[2]\!+\!X_{2}[\GREEN{2}]\phantom{{}\!+\!X_{3}[0]\!+\!X_{4}[0]}$ &                                       \\ 
        \hline
        \RED{11} & ~\BLUE{0} & ~\GREEN{1} & \multicolumn{1}{c}{0\BLUE{1}\RED{11}} & \smallmath $X_{3}[\RED{1}]\!+\!X_{4}[\RED{1}]$ & \multicolumn{1}{c}{1\GREEN{1}\RED{11}} & \smallmath $X_{1}[1]\!+\!X_{2}[\GREEN{1}]\!+\!X_{3}[\RED{1}]\!+\!X_{4}[\RED{1}]$     & \multicolumn{1}{c}{2\GREEN{1}\RED{11}} & \smallmath $X_{1}[2]\!+\!X_{2}[\GREEN{1}]\!+\!X_{3}[\RED{1}]\!+\!X_{4}[\RED{1}]$ & \multirow{8}{*}{$\frac{1}{24} P_{1}$} \\
        \RED{11} & ~\BLUE{0} & ~\GREEN{2} & \multicolumn{1}{c}{0\BLUE{2}\RED{11}} & \smallmath $X_{3}[\RED{1}]\!+\!X_{4}[\RED{1}]$ & \multicolumn{1}{c}{1\GREEN{2}\RED{11}} & \smallmath $X_{1}[1]\!+\!X_{2}[\GREEN{2}]\!+\!X_{3}[\RED{1}]\!+\!X_{4}[\RED{1}]$     & \multicolumn{1}{c}{2\GREEN{2}\RED{11}} & \smallmath $X_{1}[2]\!+\!X_{2}[\GREEN{2}]\!+\!X_{3}[\RED{1}]\!+\!X_{4}[\RED{1}]$ &                                       \\
        \RED{12} & ~\BLUE{0} & ~\GREEN{1} & \multicolumn{1}{c}{0\BLUE{1}\RED{12}} & \smallmath $X_{3}[\RED{1}]\!+\!X_{4}[\RED{2}]$ & \multicolumn{1}{c}{1\GREEN{1}\RED{12}} & \smallmath $X_{1}[1]\!+\!X_{2}[\GREEN{1}]\!+\!X_{3}[\RED{1}]\!+\!X_{4}[\RED{2}]$     & \multicolumn{1}{c}{2\GREEN{1}\RED{12}} & \smallmath $X_{1}[2]\!+\!X_{2}[\GREEN{1}]\!+\!X_{3}[\RED{1}]\!+\!X_{4}[\RED{2}]$ &                                       \\
        \RED{12} & ~\BLUE{0} & ~\GREEN{2} & \multicolumn{1}{c}{0\BLUE{2}\RED{12}} & \smallmath $X_{3}[\RED{1}]\!+\!X_{4}[\RED{2}]$ & \multicolumn{1}{c}{1\GREEN{2}\RED{12}} & \smallmath $X_{1}[1]\!+\!X_{2}[\GREEN{2}]\!+\!X_{3}[\RED{1}]\!+\!X_{4}[\RED{2}]$     & \multicolumn{1}{c}{2\GREEN{1}\RED{12}} & \smallmath $X_{1}[2]\!+\!X_{2}[\GREEN{2}]\!+\!X_{3}[\RED{1}]\!+\!X_{4}[\RED{2}]$ &                                       \\
        \RED{21} & ~\BLUE{0} & ~\GREEN{1} & \multicolumn{1}{c}{0\BLUE{1}\RED{21}} & \smallmath $X_{3}[\RED{2}]\!+\!X_{4}[\RED{1}]$ & \multicolumn{1}{c}{1\GREEN{1}\RED{21}} & \smallmath $X_{1}[1]\!+\!X_{2}[\GREEN{1}]\!+\!X_{3}[\RED{2}]\!+\!X_{4}[\RED{1}]$     & \multicolumn{1}{c}{2\GREEN{1}\RED{21}} & \smallmath $X_{1}[2]\!+\!X_{2}[\GREEN{1}]\!+\!X_{3}[\RED{2}]\!+\!X_{4}[\RED{1}]$ &                                       \\
        \RED{21} & ~\BLUE{0} & ~\GREEN{2} & \multicolumn{1}{c}{0\BLUE{2}\RED{21}} & \smallmath $X_{3}[\RED{2}]\!+\!X_{4}[\RED{1}]$ & \multicolumn{1}{c}{1\GREEN{2}\RED{21}} & \smallmath $X_{1}[1]\!+\!X_{2}[\GREEN{2}]\!+\!X_{3}[\RED{2}]\!+\!X_{4}[\RED{1}]$     & \multicolumn{1}{c}{2\GREEN{2}\RED{21}} & \smallmath $X_{1}[2]\!+\!X_{2}[\GREEN{2}]\!+\!X_{3}[\RED{2}]\!+\!X_{4}[\RED{1}]$ &                                       \\
        \RED{22} & ~\BLUE{0} & ~\GREEN{1} & \multicolumn{1}{c}{0\BLUE{1}\RED{22}} & \smallmath $X_{3}[\RED{2}]\!+\!X_{4}[\RED{2}]$ & \multicolumn{1}{c}{1\GREEN{1}\RED{22}} & \smallmath $X_{1}[1]\!+\!X_{2}[\GREEN{1}]\!+\!X_{3}[\RED{2}]\!+\!X_{4}[\RED{2}]$     & \multicolumn{1}{c}{2\GREEN{1}\RED{22}} & \smallmath $X_{1}[2]\!+\!X_{2}[\GREEN{1}]\!+\!X_{3}[\RED{2}]\!+\!X_{4}[\RED{2}]$ &                                       \\
        \RED{22} & ~\BLUE{0} & ~\GREEN{2} & \multicolumn{1}{c}{0\BLUE{2}\RED{22}} & \smallmath $X_{3}[\RED{2}]\!+\!X_{4}[\RED{2}]$ & \multicolumn{1}{c}{1\GREEN{2}\RED{22}} & \smallmath $X_{1}[1]\!+\!X_{2}[\GREEN{2}]\!+\!X_{3}[\RED{2}]\!+\!X_{4}[\RED{2}]$     & \multicolumn{1}{c}{2\GREEN{1}\RED{22}} & \smallmath $X_{1}[2]\!+\!X_{2}[\GREEN{2}]\!+\!X_{3}[\RED{2}]\!+\!X_{4}[\RED{2}]$ &                                       \\
        \bottomrule
    \end{tabular}
\end{table*}

\begin{table*}[!tbp]
    \caption{Realizations of $\key_{2} = \big(\pi, f_{U}, f^{(0)}_{S}, f^{(1)}_{S}\big)$ with $\pi = (0, 1, 2)$ and $(N, K, M) = (3, 4, 1)$ for $\varepsilon$-leaky $(W,S)$-privacy. \\
            Here, $P_{0} = \frac{1}{(2e^{-\varepsilon} + 1)^2}$, $P_{1} = \frac{4e^{-\varepsilon}}{(2e^{-\varepsilon} + 1)^2}$, and $P_{2} = \frac{4e^{-2\varepsilon}}{(2e^{-\varepsilon} + 1)^2}$.\\}
    \label{res2: ws-privacy}
    \centering
    \setlength{\tabcolsep}{3pt}
    \begin{tabular}{l@{\hspace{3pt}}l@{\hspace{1pt}}l c@{\hspace{2pt}}c c@{\hspace{2pt}}c c@{\hspace{2pt}}c c}
        \toprule
        \multirow{2}{*}{\smallmath $f^{\phantom{()}}_{U}$} & 
        \multirow{2}{*}{\smallmath $f^{(0)}_{S}$} & 
        \multirow{2}{*}{\smallmath $f^{(1)}_{S}$} & \multicolumn{2}{c}{Server $1$}                                                                           & \multicolumn{2}{c}{Server $2$}                                                                                                      & \multicolumn{2}{c}{Server $3$}                                                                                                      & \multirow{2}{*}{Prob.}                \\
        \cline{4-9}
                 &            &           & \multicolumn{1}{c}{$q_{1}$}           & $a_{1}$                                                                  & \multicolumn{1}{c}{$q_{2}$}            & $a_{2}$                                                                                    & \multicolumn{1}{c}{$q_{3}$}            & $a_{3}$                                                                                    &                                       \\
        \midrule
        \RED{00} & ~\BLUE{0} & ~\GREEN{1} & \multicolumn{1}{c}{0\BLUE{0}\RED{00}} & $\varnothing$                                                            & \multicolumn{1}{c}{1\GREEN{1}\RED{00}} & \smallmath $X_{1}[1]\!+\!X_{2}[\GREEN{1}]\phantom{{}\!+\!X_{3}[0]}$                        & \multicolumn{1}{c}{2\GREEN{1}\RED{00}} & \smallmath $X_{1}[2]\!+\!X_{2}[\GREEN{1}]\phantom{{}\!+\!X_{3}[0]}$                        & \multirow{2}{*}{$\frac{1}{24} P_{0}$} \\
        \RED{00} & ~\BLUE{0} & ~\GREEN{2} & \multicolumn{1}{c}{0\BLUE{0}\RED{00}} & $\varnothing$                                                            & \multicolumn{1}{c}{1\GREEN{2}\RED{00}} & \smallmath $X_{1}[1]\!+\!X_{2}[\GREEN{2}]\phantom{{}\!+\!X_{3}[0]}$                        & \multicolumn{1}{c}{2\GREEN{2}\RED{00}} & \smallmath $X_{1}[2]\!+\!X_{2}[\GREEN{2}]\phantom{{}\!+\!X_{3}[0]}$                        &                                       \\
        \hline
        \RED{01} & ~\BLUE{1} & ~\GREEN{0} & \multicolumn{1}{c}{0\BLUE{1}\RED{01}} & \smallmath $X_{2}[\BLUE{1}]\phantom{{}\!+\!X_{3}[0]}\!+\!X_{4}[\RED{1}]$ & \multicolumn{1}{c}{1\GREEN{0}\RED{01}} & \smallmath $X_{1}[1]\phantom{{}\!+\!X_{2}[0]\!+\!X_{3}[0]}\!+\!X_{4}[\RED{1}]$             & \multicolumn{1}{c}{2\GREEN{0}\RED{01}} & \smallmath $X_{1}[2]\phantom{{}\!+\!X_{2}[0]\!+\!X_{3}[0]}\!+\!X_{4}[\RED{1}]$             & \multirow{15}{*}{$\frac{1}{24} P_{1}$}\\
        \RED{01} & ~\BLUE{1} & ~\GREEN{1} & \multicolumn{1}{c}{0\BLUE{1}\RED{01}} & \smallmath $X_{2}[\BLUE{1}]\phantom{{}\!+\!X_{3}[0]}\!+\!X_{4}[\RED{1}]$ & \multicolumn{1}{c}{1\GREEN{1}\RED{01}} & \smallmath $X_{1}[1]\!+\!X_{2}[\GREEN{1}]\phantom{{}\!+\!X_{3}[0]}\!+\!X_{4}[\RED{1}]$     & \multicolumn{1}{c}{2\GREEN{1}\RED{01}} & \smallmath $X_{1}[2]\!+\!X_{2}[\GREEN{1}]\phantom{{}\!+\!X_{3}[0]}\!+\!X_{4}[\RED{1}]$     &                                       \\
        \RED{01} & ~\BLUE{1} & ~\GREEN{2} & \multicolumn{1}{c}{0\BLUE{1}\RED{01}} & \smallmath $X_{2}[\BLUE{1}]\phantom{{}\!+\!X_{3}[0]}\!+\!X_{4}[\RED{1}]$ & \multicolumn{1}{c}{1\GREEN{2}\RED{01}} & \smallmath $X_{1}[1]\!+\!X_{2}[\GREEN{2}]\phantom{{}\!+\!X_{3}[0]}\!+\!X_{4}[\RED{1}]$     & \multicolumn{1}{c}{2\GREEN{2}\RED{01}} & \smallmath $X_{1}[2]\!+\!X_{2}[\GREEN{2}]\phantom{{}\!+\!X_{3}[0]}\!+\!X_{4}[\RED{1}]$     &                                       \\
        \RED{01} & ~\BLUE{2} & ~\GREEN{0} & \multicolumn{1}{c}{0\BLUE{2}\RED{01}} & \smallmath $X_{2}[\BLUE{2}]\phantom{{}\!+\!X_{3}[0]}\!+\!X_{4}[\RED{1}]$ & \multicolumn{1}{c}{1\GREEN{0}\RED{01}} & \smallmath $X_{1}[1]\phantom{{}\!+\!X_{2}[0]\!+\!X_{3}[0]}\!+\!X_{4}[\RED{1}]$             & \multicolumn{1}{c}{2\GREEN{0}\RED{01}} & \smallmath $X_{1}[2]\phantom{{}\!+\!X_{2}[0]\!+\!X_{3}[0]}\!+\!X_{4}[\RED{1}]$             &                                       \\
        \RED{01} & ~\BLUE{2} & ~\GREEN{1} & \multicolumn{1}{c}{0\BLUE{2}\RED{01}} & \smallmath $X_{2}[\BLUE{2}]\phantom{{}\!+\!X_{3}[0]}\!+\!X_{4}[\RED{1}]$ & \multicolumn{1}{c}{1\GREEN{1}\RED{01}} & \smallmath $X_{1}[1]\!+\!X_{2}[\GREEN{1}]\phantom{{}\!+\!X_{3}[0]}\!+\!X_{4}[\RED{1}]$     & \multicolumn{1}{c}{2\GREEN{1}\RED{01}} & \smallmath $X_{1}[1]\!+\!X_{2}[\GREEN{1}]\phantom{{}\!+\!X_{3}[0]}\!+\!X_{4}[\RED{1}]$     &                                       \\
        \RED{01} & ~\BLUE{2} & ~\GREEN{2} & \multicolumn{1}{c}{0\BLUE{2}\RED{02}} & \smallmath $X_{2}[\BLUE{2}]\phantom{{}\!+\!X_{3}[0]}\!+\!X_{4}[\RED{1}]$ & \multicolumn{1}{c}{1\GREEN{2}\RED{02}} & \smallmath $X_{1}[1]\!+\!X_{2}[\GREEN{2}]\phantom{{}\!+\!X_{3}[0]}\!+\!X_{4}[\RED{1}]$     & \multicolumn{1}{c}{2\GREEN{0}\RED{02}} & \smallmath $X_{1}[2]\!+\!X_{2}[\GREEN{2}]\phantom{{}\!+\!X_{3}[0]}\!+\!X_{4}[\RED{1}]$     &                                       \\
        \RED{02} & ~\BLUE{1} & ~\GREEN{0} & \multicolumn{1}{c}{0\BLUE{1}\RED{02}} & \smallmath $X_{2}[\BLUE{1}]\phantom{{}\!+\!X_{3}[0]}\!+\!X_{4}[\RED{2}]$ & \multicolumn{1}{c}{1\GREEN{0}\RED{02}} & \smallmath $X_{1}[1]\phantom{{}\!+\!X_{2}[0]\!+\!X_{3}[0]}\!+\!X_{4}[\RED{2}]$             & \multicolumn{1}{c}{2\GREEN{0}\RED{02}} & \smallmath $X_{1}[2]\phantom{{}\!+\!X_{2}[0]\!+\!X_{3}[0]}\!+\!X_{4}[\RED{2}]$             &                                       \\
        \RED{02} & ~\BLUE{1} & ~\GREEN{1} & \multicolumn{1}{c}{0\BLUE{1}\RED{02}} & \smallmath $X_{2}[\BLUE{1}]\phantom{{}\!+\!X_{3}[0]}\!+\!X_{4}[\RED{2}]$ & \multicolumn{1}{c}{1\GREEN{1}\RED{02}} & \smallmath $X_{1}[1]\!+\!X_{2}[\GREEN{1}]\phantom{{}\!+\!X_{3}[0]}\!+\!X_{4}[\RED{2}]$     & \multicolumn{1}{c}{2\GREEN{1}\RED{02}} & \smallmath $X_{1}[1]\!+\!X_{2}[\GREEN{1}]\phantom{{}\!+\!X_{3}[0]}\!+\!X_{4}[\RED{2}]$     &                                       \\
        \RED{02} & ~\BLUE{1} & ~\GREEN{2} & \multicolumn{1}{c}{0\BLUE{1}\RED{02}} & \smallmath $X_{2}[\BLUE{1}]\phantom{{}\!+\!X_{3}[0]}\!+\!X_{4}[\RED{2}]$ & \multicolumn{1}{c}{1\GREEN{2}\RED{02}} & \smallmath $X_{1}[1]\!+\!X_{2}[\GREEN{2}]\phantom{{}\!+\!X_{3}[0]}\!+\!X_{4}[\RED{2}]$     & \multicolumn{1}{c}{2\GREEN{2}\RED{02}} & \smallmath $X_{1}[2]\!+\!X_{2}[\GREEN{2}]\phantom{{}\!+\!X_{3}[0]}\!+\!X_{4}[\RED{2}]$     &                                       \\            
        \RED{02} & ~\BLUE{2} & ~\GREEN{0} & \multicolumn{1}{c}{0\BLUE{1}\RED{02}} & \smallmath $X_{2}[\BLUE{2}]\phantom{{}\!+\!X_{3}[0]}\!+\!X_{4}[\RED{2}]$ & \multicolumn{1}{c}{1\GREEN{0}\RED{02}} & \smallmath $X_{1}[1]\phantom{{}\!+\!X_{2}[0]\!+\!X_{3}[0]}\!+\!X_{4}[\RED{2}]$             & \multicolumn{1}{c}{2\GREEN{0}\RED{02}} & \smallmath $X_{1}[2]\phantom{{}\!+\!X_{2}[0]\!+\!X_{3}[0]}\!+\!X_{4}[\RED{2}]$             &                                       \\
        \RED{02} & ~\BLUE{2} & ~\GREEN{1} & \multicolumn{1}{c}{0\BLUE{2}\RED{02}} & \smallmath $X_{2}[\BLUE{2}]\phantom{{}\!+\!X_{3}[0]}\!+\!X_{4}[\RED{2}]$ & \multicolumn{1}{c}{1\GREEN{1}\RED{02}} & \smallmath $X_{1}[1]\!+\!X_{2}[\GREEN{1}]\phantom{{}\!+\!X_{3}[0]}\!+\!X_{4}[\RED{2}]$     & \multicolumn{1}{c}{2\GREEN{1}\RED{02}} & \smallmath $X_{1}[2]\!+\!X_{2}[\GREEN{1}]\phantom{{}\!+\!X_{3}[0]}\!+\!X_{4}[\RED{2}]$     &                                       \\                            
        \RED{02} & ~\BLUE{2} & ~\GREEN{2} & \multicolumn{1}{c}{0\BLUE{2}\RED{02}} & \smallmath $X_{2}[\BLUE{2}]\phantom{{}\!+\!X_{3}[0]}\!+\!X_{4}[\RED{2}]$ & \multicolumn{1}{c}{1\GREEN{2}\RED{02}} & \smallmath $X_{1}[1]\!+\!X_{2}[\GREEN{2}]\phantom{{}\!+\!X_{3}[0]}\!+\!X_{4}[\RED{2}]$     & \multicolumn{1}{c}{2\GREEN{2}\RED{02}} & \smallmath $X_{1}[2]\!+\!X_{2}[\GREEN{2}]\phantom{{}\!+\!X_{3}[0]}\!+\!X_{4}[\RED{2}]$     &                                       \\
        \RED{10} & ~\BLUE{1} & ~\GREEN{0} & \multicolumn{1}{c}{0\BLUE{1}\RED{10}} & \smallmath $X_{2}[\BLUE{1}]\!+\!X_{3}[\RED{1}]\phantom{{}\!+\!X_{4}[0]}$ & \multicolumn{1}{c}{1\GREEN{0}\RED{10}} & \smallmath $X_{1}[1]\phantom{{}\!+\!X_{2}[0]}\!+\!X_{3}[\RED{1}]\phantom{{}\!+\!X_{4}[0]}$ & \multicolumn{1}{c}{2\GREEN{0}\RED{10}} & \smallmath $X_{1}[2]\phantom{{}\!+\!X_{2}[0]}\!+\!X_{3}[\RED{1}]\phantom{{}\!+\!X_{4}[0]}$ &                                       \\
                 & ~$\vdots$  &           & \multicolumn{1}{c}{$\vdots$}          & $\vdots$                                                                 & \multicolumn{1}{c}{$\vdots$}           & $\vdots$                                                                                   & \multicolumn{1}{c}{$\vdots$}           & $\vdots$                                                                                   &                                       \\
        \RED{20} & ~\BLUE{2} & ~\GREEN{2} & \multicolumn{1}{c}{0\BLUE{2}\RED{20}} & \smallmath $X_{2}[\BLUE{2}]\!+\!X_{3}[\RED{2}]\phantom{{}\!+\!X_{4}[0]}$ & \multicolumn{1}{c}{1\GREEN{2}\RED{20}} & \smallmath $X_{1}[1]\!+\!X_{2}[\GREEN{2}]\!+\!X_{3}[\RED{2}]\phantom{{}\!+\!X_{4}[0]}$     & \multicolumn{1}{c}{2\GREEN{2}\RED{20}} & \smallmath $X_{1}[2]\!+\!X_{2}[\GREEN{2}]\!+\!X_{3}[\RED{2}]\phantom{{}\!+\!X_{4}[0]}$     &                                       \\
        \hline
        \RED{11} & ~\BLUE{0} & ~\GREEN{0} & \multicolumn{1}{c}{0\BLUE{0}\RED{11}} & \smallmath $\phantom{X_{2}[0]\!+\!{}}X_{3}[\RED{1}]\!+\!X_{4}[\RED{1}]$  & \multicolumn{1}{c}{1\GREEN{0}\RED{11}} & \smallmath $X_{1}[1]\phantom{{}\!+\!X_{2}[0]}\!+\!X_{3}[\RED{1}]\!+\!X_{4}[\RED{1}]$       & \multicolumn{1}{c}{2\GREEN{0}\RED{11}} & \smallmath $X_{1}[2]\phantom{{}\!+\!X_{2}[0]}\!+\!X_{3}[\RED{1}]\!+\!X_{4}[\RED{1}]$       & \multirow{12}{*}{$\frac{1}{24} P_{2}$}\\  
        \RED{11} & ~\BLUE{0} & ~\GREEN{1} & \multicolumn{1}{c}{0\BLUE{0}\RED{11}} & \smallmath $\phantom{X_{2}[0]\!+\!{}}X_{3}[\RED{1}]\!+\!X_{4}[\RED{1}]$  & \multicolumn{1}{c}{1\GREEN{1}\RED{11}} & \smallmath $X_{1}[1]\!+\!X_{2}[\GREEN{1}]\!+\!X_{3}[\RED{1}]\!+\!X_{4}[\RED{1}]$           & \multicolumn{1}{c}{2\GREEN{1}\RED{11}} & \smallmath $X_{1}[2]\!+\!X_{2}[\GREEN{1}]\!+\!X_{3}[\RED{1}]\!+\!X_{4}[\RED{1}]$           &                                       \\  
        \RED{11} & ~\BLUE{0} & ~\GREEN{2} & \multicolumn{1}{c}{0\BLUE{0}\RED{11}} & \smallmath $\phantom{X_{2}[0]\!+\!{}}X_{3}[\RED{1}]\!+\!X_{4}[\RED{1}]$  & \multicolumn{1}{c}{1\GREEN{2}\RED{11}} & \smallmath $X_{1}[1]\!+\!X_{2}[\GREEN{2}]\!+\!X_{3}[\RED{1}]\!+\!X_{4}[\RED{1}]$           & \multicolumn{1}{c}{2\GREEN{2}\RED{11}} & \smallmath $X_{1}[2]\!+\!X_{2}[\GREEN{2}]\!+\!X_{3}[\RED{1}]\!+\!X_{4}[\RED{1}]$           &                                       \\
        \RED{11} & ~\BLUE{1} & ~\GREEN{0} & \multicolumn{1}{c}{0\BLUE{1}\RED{11}} & \smallmath $X_{2}[\BLUE{1}]\!+\!X_{3}[\RED{1}]\!+\!X_{4}[\RED{1}]$       & \multicolumn{1}{c}{1\GREEN{0}\RED{11}} & \smallmath $X_{1}[1]\phantom{{}\!+\!X_{2}[0]}\!+\!X_{3}[\RED{1}]\!+\!X_{4}[\RED{1}]$       & \multicolumn{1}{c}{2\GREEN{0}\RED{11}} & \smallmath $X_{1}[2]\phantom{{}\!+\!X_{2}[0]}\!+\!X_{3}[\RED{1}]\!+\!X_{4}[\RED{1}]$       &                                       \\
        \RED{11} & ~\BLUE{1} & ~\GREEN{1} & \multicolumn{1}{c}{0\BLUE{1}\RED{11}} & \smallmath $X_{2}[\BLUE{1}]\!+\!X_{3}[\RED{1}]\!+\!X_{4}[\RED{1}]$       & \multicolumn{1}{c}{1\GREEN{1}\RED{11}} & \smallmath $X_{1}[1]\!+\!X_{2}[\GREEN{1}]\!+\!X_{3}[\RED{1}]\!+\!X_{4}[\RED{1}]$           & \multicolumn{1}{c}{2\GREEN{1}\RED{11}} & \smallmath $X_{1}[2]\!+\!X_{2}[\GREEN{1}]\!+\!X_{3}[\RED{1}]\!+\!X_{4}[\RED{1}]$           &                                       \\
        \RED{11} & ~\BLUE{1} & ~\GREEN{2} & \multicolumn{1}{c}{0\BLUE{1}\RED{11}} & \smallmath $X_{2}[\BLUE{1}]\!+\!X_{3}[\RED{1}]\!+\!X_{4}[\RED{1}]$       & \multicolumn{1}{c}{1\GREEN{2}\RED{11}} & \smallmath $X_{1}[1]\!+\!X_{2}[\GREEN{2}]\!+\!X_{3}[\RED{1}]\!+\!X_{4}[\RED{1}]$           & \multicolumn{1}{c}{2\GREEN{2}\RED{11}} & \smallmath $X_{1}[2]\!+\!X_{2}[\GREEN{2}]\!+\!X_{3}[\RED{1}]\!+\!X_{4}[\RED{1}]$           &                                       \\
        \RED{11} & ~\BLUE{2} & ~\GREEN{0} & \multicolumn{1}{c}{0\BLUE{2}\RED{11}} & \smallmath $X_{2}[\BLUE{2}]\!+\!X_{3}[\RED{1}]\!+\!X_{4}[\RED{1}]$       & \multicolumn{1}{c}{1\GREEN{0}\RED{11}} & \smallmath $X_{1}[1]\phantom{{}\!+\!X_{2}[0]}\!+\!X_{3}[\RED{1}]\!+\!X_{4}[\RED{1}]$       & \multicolumn{1}{c}{2\GREEN{0}\RED{11}} & \smallmath $X_{1}[2]\phantom{{}\!+\!X_{2}[0]}\!+\!X_{3}[\RED{1}]\!+\!X_{4}[\RED{1}]$       &                                       \\
        \RED{11} & ~\BLUE{2} & ~\GREEN{1} & \multicolumn{1}{c}{0\BLUE{2}\RED{11}} & \smallmath $X_{2}[\BLUE{2}]\!+\!X_{3}[\RED{1}]\!+\!X_{4}[\RED{1}]$       & \multicolumn{1}{c}{1\GREEN{1}\RED{11}} & \smallmath $X_{1}[1]\!+\!X_{2}[\GREEN{1}]\!+\!X_{3}[\RED{1}]\!+\!X_{4}[\RED{1}]$           & \multicolumn{1}{c}{2\GREEN{1}\RED{11}} & \smallmath $X_{1}[2]\!+\!X_{2}[\GREEN{1}]\!+\!X_{3}[\RED{1}]\!+\!X_{4}[\RED{1}]$           &                                       \\
        \RED{11} & ~\BLUE{2} & ~\GREEN{2} & \multicolumn{1}{c}{0\BLUE{2}\RED{11}} & \smallmath $X_{2}[\BLUE{2}]\!+\!X_{3}[\RED{1}]\!+\!X_{4}[\RED{1}]$       & \multicolumn{1}{c}{1\GREEN{2}\RED{11}} & \smallmath $X_{1}[1]\!+\!X_{2}[\GREEN{2}]\!+\!X_{3}[\RED{1}]\!+\!X_{4}[\RED{1}]$           & \multicolumn{1}{c}{2\GREEN{2}\RED{11}} & \smallmath $X_{1}[2]\!+\!X_{2}[\GREEN{2}]\!+\!X_{3}[\RED{1}]\!+\!X_{4}[\RED{1}]$           &                                       \\
        \RED{12} & ~\BLUE{0} & ~\GREEN{0} & \multicolumn{1}{c}{0\BLUE{1}\RED{12}} & \smallmath $\phantom{X_{2}[0]\!+\!{}}X_{3}[\RED{1}]\!+\!X_{4}[\RED{2}]$  & \multicolumn{1}{c}{1\GREEN{0}\RED{12}} & \smallmath $X_{1}[1]\phantom{{}\!+\!X_{2}[0]}\!+\!X_{3}[\RED{1}]\!+\!X_{4}[\RED{2}]$       & \multicolumn{1}{c}{2\GREEN{0}\RED{12}} & \smallmath $X_{1}[2]\phantom{{}\!+\!X_{2}[0]}\!+\!X_{3}[\RED{1}]\!+\!X_{4}[\RED{2}]$       &                                       \\
                 & ~$\vdots$ &            & \multicolumn{1}{c}{$\vdots$}          & $\vdots$                                                                 & \multicolumn{1}{c}{$\vdots$}           & $\vdots$                                                                                   & \multicolumn{1}{c}{$\vdots$}           & $\vdots$                                                                                   &                                       \\
        \RED{22} & ~\BLUE{2} & ~\GREEN{2} & \multicolumn{1}{c}{0\BLUE{2}\RED{22}} & \smallmath $X_{2}[\BLUE{2}]\!+\!X_{3}[\RED{2}]\!+\!X_{4}[\RED{2}]$       & \multicolumn{1}{c}{1\GREEN{2}\RED{22}} & \smallmath $X_{1}[1]\!+\!X_{2}[\GREEN{2}]\!+\!X_{3}[\RED{2}]\!+\!X_{4}[\RED{2}]$           & \multicolumn{1}{c}{2\GREEN{2}\RED{22}} & \smallmath $X_{1}[2]\!+\!X_{2}[\GREEN{2}]\!+\!X_{3}[\RED{2}]\!+\!X_{4}[\RED{2}]$           &                                       \\
        \bottomrule
        \end{tabular}
\end{table*}

We provide another example with $N =  3$, $K = 4$, $M = 1$, and $L = 2$ to further illustrate the proposed schemes for either $\varepsilon$-leaky $W$-privacy or $(W, S)$-privacy.
In this case, we have $X_{i} = [X_{i}[1]]$ for $i \in [4]$.
Assume that the user wants to retrieve message $X_{1}$ with side information $X_{2}$~(i.e., $W = 1$, $S = \Set{2}$, $U = \Set{3, 4}$, and $g = 2$).
Note that $\perm$ takes $4! = 24$ possible permutations with equal probability in this example.

For the $\varepsilon$-leaky $W$-privacy scheme, we only present the realizations of the random pattern $\key_{1} = \big(\pi, f_{U}, f^{(0)}_{S}, f^{(1)}_{S} \big)$ with $\pi = (0, 1, 2)$.
The possible realizations of $\key_{1}$ are listed in Table~\ref{res2: w-privacy}.
Notice that $K$ is divisible by $M + 1$ in this example.
Thus, $f_{U}$ always has a Hamming weight of $K - M - 1 = 2$ if $k = \lceil g \rceil - 1 = 1$, and $f^{(0)}_{S}$ is always the all-zero vector.

For the $\varepsilon$-leaky $(W, S)$-privacy scheme, we also only present $\key_{2} = \big(\pi, f_{U}, f^{(0)}_{S}, f^{(1)}_{S}\big)$ with $\pi = (0, 1, 2)$.
The possible realizations of $\key_{2}$ are listed in Table~\ref{res2: ws-privacy}.
If $f_{U}$ has a Hamming weight of $1$, $(f^{(0)}_{S}, f^{(1)}_{S})$ have Hamming weights of either $(1,0)$ or $(1,1)$, to ensure that the overall Hamming weight is at least $M + 1 = 2$.
If $\ptrn_{U}$ has a Hamming weight of $2$, $(f^{(0)}_{S}, f^{(1)}_{S})$ can have any combination of Hamming weights, i.e., $(0,0)$, $(0,1)$, $(1,0)$ or $(1,1)$.

\section{Proof of Theorem~\ref{thm: leakage W-privacy}} \label{app: leakage WS-privacy}

Fix the download cost $D$ and the number of servers $N$.
According to Theorem~\ref{thm: W-privacy}, there exists an $\varepsilon$-leaky $W$-privacy scheme with download cost $D$, when the leakage ratio exponent $\varepsilon$ satisfies
\begin{equation*}
    D = 1 + \frac{1}{N - 1} \cdot \bigg( 1 - \frac{1}{\sum_{k = 0}^{\lceil g \rceil - 1} \binom{g - 1}{k} (N - 1)^{k} \ee^{-k\varepsilon}} \bigg),
\end{equation*}
where $g = \frac{K}{M + 1}$ and $\binom{g - 1}{k}$ is the generalized binomial coefficient.

Define a constant $C = - \ln \big(1 - (N - 1)(D - 1)\big)$ and we know $C \geq 0$ as $D \leq 1 + \frac{1}{N - 1}$.
Then, we can rewrite the above equality as
\begin{equation*}
    \ee^{C} = \sum_{k = 0}^{\lceil g \rceil - 1} \binom{g - 1}{k} (N - 1)^{k} \ee^{-k\varepsilon}.
\end{equation*}
Notice that $\binom{g - 1}{k} \leq \binom{\lceil g \rceil - 1}{k}$ as $g \geq k$.
We have
\begin{equation*}
    \ee^{C} 
    \leq \sum_{k = 0}^{\lceil g \rceil - 1} \binom{\lceil g \rceil - 1}{k} (N - 1)^{k} \ee^{-k\varepsilon} 
    = \big(1 + (N - 1)\ee^{-\varepsilon}\big)^{\lceil g \rceil - 1}.
\end{equation*}
It follows that
\begin{equation*}
    C \leq \big(\lceil g \rceil - 1\big) \ln \big(1 + (N - 1)\ee^{-\varepsilon}\big).
\end{equation*}

In addition, using the fact that $\ln(x + 1) \leq x$ for any $x > 0$, we have
\begin{equation*}
    C \leq \big(\lceil g \rceil - 1\big) (N - 1)\ee^{-\varepsilon},
\end{equation*}
which implies
\begin{equation*}
    \ee^{\varepsilon} \leq \frac{\big(\lceil g \rceil - 1\big)(N - 1)}{C}.
\end{equation*}
As a result, we derive the upper bound on the leakage ratio exponent as
\begin{equation*}
    \varepsilon \leq \ln \big(\big(\lceil g \rceil - 1\big)(N - 1)\big) - \ln C = \bigO(\log g).
\end{equation*}
This gives the desired scaling behavior of
\begin{equation*}
    \varepsilon = \bigO\Big(\log \frac{K}{M + 1}\Big).    
\end{equation*}

\section{Proof of Theorem~\ref{thm: leakage WS-privacy}} \label{app: leakage W-privacy}

Fix the download cost $D$ and the number of servers $N$.
According to Theorem~\ref{thm: WS-privacy}, there exists an $\varepsilon$-leaky $(W, S)$-privacy scheme with download cost $D$, when the leakage ratio exponent $\varepsilon$ satisfies 
\begin{equation*}
    D = 1 + \frac{1}{N - 1} \cdot \Big( 1 - \frac{1}{((N - 1)\ee^{-\varepsilon} + 1)^{K - 2}} \Big).
\end{equation*}

Define a constant $C = - \ln \big(1 - (N - 1)(D - 1)\big)$ and we know $C \geq 0$ as $D \leq 1 + \frac{1}{N - 1}$.
Then, we can rewrite the above equality as
\begin{equation*}
    C = (K - 2) \ln \big((N - 1)\ee^{-\varepsilon} + 1\big).
\end{equation*}
Using the fact that $\ln(x + 1) \leq x$ for any $x > 0$, we have
\begin{equation*}
    C \leq (K - 2) (N - 1)\ee^{-\varepsilon},
\end{equation*}
which implies
\begin{equation*}
    \ee^{\varepsilon} \leq \frac{(K - 2)(N - 1)}{C}.
\end{equation*}
As a result, we derive the upper bound on the leakage ratio exponent
\begin{equation*}
    \varepsilon \leq \ln ((K - 2)(N - 1)) - \ln C = \bigO(\log K).
\end{equation*}
This completes the proof.

\section{Proof of Lemma~\ref{lem: W-privacy}} \label{app: W-privacy}

Observe that server $n$ generates its answer solely based on the query and database, the conditional distribution of $\ans{\dmd, \si}_{n}$ given $\que{\dmd,\si}_{n}$ does not depend on the demand index $W$ and the side information index set $S$.
It suffices to show that for any server $n \in [N]$ and any two demand indices $W, W' \in [K]$,
\begin{equation*}
    \frac{\bigprob{\que{\dmd,\si}_{n} \mid W}}{\bigprob{\que{\dmd,\si}_{n} \mid W'}} \leq \ee^{\varepsilon}.
\end{equation*}

To evaluate the privacy leakage, we analyze the probability of observing a specific query $q \in \qspace$ at server $n \in [N]$.
Recall that $g = \frac{K}{M + 1}$ and the size of each query $|\que{\dmd,\si}_{n}|$ is either $0$, or $K$, or $k(M + 1) $ for some $k \in \big[ \lceil g \rceil - 1 \big]$.

Let $w(q)$ be the Hamming weight of the query $q \in \qspace$.
Since the random bijection $\pi$ is uniform, the server index $n$ is uniformly distributed over $[N]$.
We can expand the probability 
\begin{equation*}
    \begin{aligned}
        \bigprob{\que{\dmd,\si}_{n} = q \mid W}
        ={}& \frac{1}{N} \sum_{S \subseteq [K] \setminus \{W\}} \Pr[S \mid n] \cdot \bigprob{\que{\dmd,\si}_{n} = q \mid W, S, n} \\
        ={}& \frac{1}{N} \cdot \frac{1}{\binom{K - 1}{M}} \sum_{S \subseteq [K] \setminus \{W\}} \bigprob{\que{\dmd,\si}_{n} = q \mid W, S, n} \\
    \end{aligned}
\end{equation*}
Now, we analyze it by three cases according to the Hamming weight of the query $q \in \qspace$.

\textbf{Case 1: $w(q) = 0$.}

This case is straightforward.
All indices in the query are $0$, indicating that the server $n$ must be the inference server.
This indicates that there are $\binom{K - 1}{M}$ possible side information index sets $S$, and the vectors $\ptrn_{U}$ and $\ptrn_{S}^{(0)}$ must be all-zero vectors.
Thus, the probability $\bigprob{\que{\dmd,\si}_{n} = q \mid W}$ is given by
\begin{equation}
    \bigprob{\que{\dmd,\si}_{n} = q \mid W} 
    = \frac{1}{N} \cdot \frac{\binom{K - 1}{M}}{\binom{K - 1}{M}} \cdot \bigprob{\que{\dmd,\si}_{n} = q \mid W, S, n}
    = \frac{1}{N} \cdot P_{0}.
\end{equation}

\textbf{Case 2: $w(q) = k(M + 1)$ for some $k \in [\lceil g \rceil- 1]$.}

If the server $n$ is the inference server, then $\pi(n) = 0$.
Thus, the Hamming weight of the vector $\ptrn_{U}$ and $\ptrn_{S}^{(0)}$ is respectively $\ell_{k}$ and $k(M + 1) - \ell_{k}$, where $\ell_{k} = \min \Set{k(M + 1), K - M - 1}$.
This yields that there are $\binom{k(M + 1)}{k(M + 1) - \ell_{k}}\binom{K - k(M + 1) - 1}{M - k(M + 1) + \ell_{k}}$ possible side information index sets $S$.
It follows that $\bigprob{\que{\dmd,\si}_{n} = q \mid W}$ can be expressed as
\begin{equation} \label{eq: p_k}
    \begin{aligned}
        &\bigprob{\que{\dmd,\si}_{n} = q \mid W} \\
        ={}& \frac{1}{N} \cdot  \frac{\binom{k(M + 1)}{k(M + 1) - \ell_{k}}\binom{K - k(M + 1) - 1}{M - k(M + 1) + \ell_{k}}}{\binom{K - 1}{M}} \cdot \bigprob{\que{\dmd,\si}_{n} = q \mid W, S, n} \\
        ={}& \frac{1}{N} \cdot  \frac{\binom{k(M + 1)}{\ell_{k}}\binom{K - k(M + 1) - 1}{M - k(M + 1) + \ell_{k}}}{\binom{K - 1}{M}} \cdot P_{k} \cdot \frac{1}{\binom{M}{k(M + 1) - \ell_{k}} \binom{K - M - 1}{\ell_{k}} \cdot (N - 1)^{k(M + 1)}}. \\
    \end{aligned}
\end{equation}

Conversely, we consider the case where $n$ is not the inference server, i.e., $\pi(n) \neq 0$.
Thus, the Hamming weight of the vector $\ptrn_{U}$ and $\ptrn_{S}^{(1)}$ is respectively $(k - 1)(M + 1) - 1$ and $M$.
Accordingly, there are $\binom{k(M + 1) - 1}{M}$ possible side information index sets $S$.
It follows that $\bigprob{\que{\dmd,\si}_{n} = q \mid W}$ can be expressed as
\begin{equation}\label{eq: p_k - 1}
    \begin{aligned}
        &\bigprob{\que{\dmd,\si}_{n} = q \mid W} \\
        ={}& \frac{1}{N} \cdot \frac{\binom{k(M + 1) - 1}{M}}{\binom{K - 1}{M}} \cdot \bigprob{\que{\dmd,\si}_{n} = q \mid W, S, n} \\
        ={}& \frac{1}{N} \cdot \frac{\binom{k(M + 1) - 1}{M}}{\binom{K - 1}{M}} \cdot P_{k - 1} \cdot \frac{1}{\binom{K - M - 1}{(k - 1)(M + 1)} \cdot (N - 1)^{k(M + 1) - 1}}.
    \end{aligned}
\end{equation}

\textbf{Case 3: $w(q) = K$.}

In this case, all indices in the query are non-zero, indicating that the server $n$ must be a non-inference server.
Thus, $\bigprob{\que{\dmd,\si}_{n} = q \mid W}$ can be expressed as
\begin{equation}
    \begin{aligned}
        \bigprob{\que{\dmd,\si}_{n} = q \mid W}
        =&{} \frac{1}{N} \cdot \frac{\binom{K}{M}}{\binom{K - 1}{M}} \cdot \bigprob{\que{\dmd,\si}_{n} = q \mid W, S, n} \\
        =&{} \frac{1}{N} \cdot \frac{\binom{K}{M}}{\binom{K - 1}{M}} \cdot P_{\lceil g \rceil- 1} \cdot \frac{1}{(N - 1)^{K - 1}}.
    \end{aligned}
\end{equation}

For any two demand indices $W, W'\in[K]$ and any realization $q \in \qspace$ of the query at server $n \in [N]$, the conditional probability
$\bigprob{\que{\dmd,\si}_n = q \mid W}$ can only arise from the exactly one equation in Cases 1 and 3, or from two equations \eqref{eq: p_k} and \eqref{eq: p_k - 1} in Case 2.
Therefore, for any $q \in \qspace$, the ratio
\begin{equation}
    \frac{\bigprob{\que{\dmd,\si}_{n} = q \mid W}}{\bigprob{\que{\dmd,\si}_{n} = q \mid W'}}
\end{equation}
can take only three possible values: either $1$, or
\begin{equation} \label{P_k vs P_k - 1}
    \frac{\binom{k(M + 1)}{\ell_{k}}\binom{K - k(M + 1) - 1}{M - k(M + 1) + \ell_{k}}}{\binom{k(M + 1) - 1}{M}} \cdot \frac{\binom{K - M - 1}{(k - 1)(M + 1)}}{\binom{M}{k(M + 1) - \ell_{k}} \binom{K - M - 1}{\ell_{k}} \cdot (N - 1)} \cdot \frac{P_{k}}{P_{k - 1}}
\end{equation}
or its reciprocal for some $k \in [\lceil g \rceil - 1]$.

By substituting the values of $P_{k}$ and $P_{k - 1}$, we have $P_{k} / P_{k - 1} = \frac{K - k(M + 1)}{k(M + 1)} \cdot (N - 1) \ee^{-\varepsilon}$.
When $k \in [\lceil g \rceil - 2]$, we have $\ell_{k} = k(M + 1)$, and thus Eq.~(\ref{P_k vs P_k - 1}) can be simplified as
\begin{equation}
    \frac{\binom{K - k(M + 1) - 1}{M}}{\binom{k(M + 1) - 1}{M}} \cdot \frac{\binom{K - M - 1}{(k - 1)(M + 1)}}{\binom{K - M - 1}{k(M + 1)} \cdot (N - 1)} \cdot \frac{K - k(M + 1)}{k(M + 1) \cdot (N - 1)} \cdot (N - 1) \ee^{-\varepsilon}
    = \ee^{-\varepsilon}.
\end{equation}
When $k = \lceil g \rceil - 1$, we have $\ell_{k} = K - M - 1$, and thus Eq.~(\ref{P_k vs P_k - 1}) can be simplified as
\begin{equation}
    \frac{\binom{k(M + 1)}{K - M - 1}}{\binom{k(M + 1) - 1}{M}} \cdot \frac{\binom{K - M - 1}{(k - 1)(M + 1)}}{\binom{M}{(k + 1)(M + 1) - K}} \cdot \frac{K - k(M + 1)}{k(M + 1) \cdot (N - 1)} \cdot (N - 1) \ee^{-\varepsilon}
    = \ee^{-\varepsilon}.
\end{equation}
Consequently, the ratio in Eq.~(\ref{P_k vs P_k - 1}) is always equal to $\ee^{-\varepsilon}$ for $k \in [\lceil g \rceil - 1]$.
This completes the proof of $\varepsilon$-leaky $W$-privacy for the scheme under the pattern $\key_{1}$.

\section{Proof of Lemma~\ref{lem: WS-privacy}} \label{app: WS-privacy}

Observe that server $n$ generates its answer solely based on the query and database, the conditional distribution of $\ans{\dmd, \si}_{n}$ given $\que{\dmd,\si}_{n}$ does not depend on the demand index $W$ and the side information index set $S$.
It suffices to show that for any server $n \in [N]$, any two demand indices $W, W' \in [K]$ and any two side information index sets $S \subseteq [K] \setminus \{W\}$ and $S' \subseteq [K] \setminus \{W'\}$ with $|S| = |S'| = M$,
\begin{equation*}
    \frac{\bigprob{\que{\dmd,\si}_{n} \mid W, S}}{\bigprob{\que{\dmd,\si}_{n} \mid W', S'}} \leq \ee^{\varepsilon}.
\end{equation*}

To evaluate the privacy leakage, we analyze the probability of observing a specific query $q \in \qspace$ at server $n \in [N]$.
Let $\ell$ and $s_{j}$ ($j \in [0: 1]$) be the Hamming weight of the vector $\ptrn_{U}$ and $\ptrn_{S}^{(j)}$, where $j = 0$ if $\pi(n) = 0$ and $j = 1$ otherwise.
Recall that $r = (N - 1)\ee^{-\varepsilon}$ and the probability that the vector $\ptrn_{U}$ has Hamming weight $\ell$ is
\begin{equation*}
    P_{\ell} = \frac{\binom{K - 2}{\ell} \cdot r^{\ell}}{(r + 1)^{K - 2}},
\end{equation*}
and
\begin{equation*}
    \Pr[ s_{j} \mid j, \ell] = \frac{r^{\ell + s_{j} + j} + (-1)^{\ell + s_{j} + j} \cdot r}{(r + 1) \cdot r^{\ell + j}}. 
\end{equation*}
Since the random bijection $\pi$ is uniform, the server index $n$ is uniformly distributed over $[N]$.
Besides, the Hamming weight $\ell$ is independent of the server index $n$, and the index $j$ is totally determined by $\pi(n)$.
We can expand the probability
\begin{equation*}
    \begin{aligned}
        \bigprob{\que{\dmd,\si}_{n} = q \mid W, S} 
        ={}& \frac{1}{N} \cdot \sum_{\ell = 0}^{K - 2} \sum_{s_{j} = 0}^{1} P_{\ell} \cdot \Pr[s_{j} \mid n, \ell] \cdot \bigprob{\que{\dmd,\si}_{n} = q \mid W, S, n, \ell, s_{j}} \\
        ={}& \frac{1}{N} \cdot \sum_{\ell = 0}^{K - 2} \sum_{s_{j} = 0}^{1} P_{\ell} \cdot \Pr[s_{j} \mid j, \ell] \cdot \bigprob{\que{\dmd,\si}_{n} = q \mid W, S, n, \ell, s_{j}}
    \end{aligned}
\end{equation*}

Now, we analyze it by three cases according to the Hamming weight of the query $q \in \qspace$, denoted by $w(q)$.

\textbf{Case 1: $w(q) = 0$.}

This case is straightforward.
All indices in the query are $0$, indicating that the server $n$ must be the inference server and the vectors $\ptrn_{U}$ and $\ptrn_{S}^{(0)}$ must be all-zero vectors.
The probability $\bigprob{\que{\dmd,\si}_{n} = q \mid W, S}$ is given by
\begin{equation}
    \bigprob{\que{\dmd,\si}_{n} = q \mid W, S} 
    = \frac{1}{N} \cdot P_{0} \cdot 1
    = \frac{1}{N} \cdot P_{0}.
\end{equation}

\textbf{Case 2: $w(q) = l \in [K - 1]$.}

If the server $n$ is the inference server, then $\pi(n) = 0$.
Since $M = 1$, the Hamming weight of the vectors $\ptrn_{U}$ and $\ptrn_{S}^{(0)}$ is either $(l, 0)$ or $(l - 1, 1)$.
When the pair $(\ptrn_{U}, \ptrn_{S}^{(0)})$ has Hamming weight $(l, 0)$, $\bigprob{\que{\dmd,\si}_{n} = q \mid W, S}$ can be expressed as
\begin{equation}
    \begin{aligned} 
        \bigprob{\que{\dmd,\si}_{n} = q \mid W, S}
        ={}& \frac{1}{N} \cdot P_{l} \cdot \frac{r^{l} + (-1)^{l} \cdot r}{(r + 1) \cdot r^{l}} \cdot \frac{1}{\binom{K - 2}{l} \cdot (N - 1)^{l}} \\
        ={}& \frac{1}{N(N - 1)^{l}} \cdot  \frac{r^{l} + (-1)^{l} \cdot r}{(r + 1)^{K - 1}},
    \end{aligned}
\end{equation}
and when the pair $(\ptrn_{U}, \ptrn_{S}^{(0)})$ has Hamming weight $(l - 1, 1)$, $\bigprob{\que{\dmd,\si}_{n} = q \mid W, S}$ can be expressed as
\begin{equation}
    \begin{aligned} 
        \bigprob{\que{\dmd,\si}_{n} = q \mid W, S}
        ={}& \frac{1}{N} \cdot P_{l - 1} \cdot \frac{r^{(l - 1) + 1} + (-1)^{(l - 1) + 1} \cdot r}{(r + 1) \cdot r^{l - 1}} \cdot \frac{1}{\binom{K - 2}{l - 1} \cdot (N - 1)^{l}} \\
        ={}& \frac{1}{N} \cdot P_{l - 1} \cdot \frac{r^{l} + (-1)^{l} \cdot r}{(r + 1) \cdot r^{l - 1}} \cdot \frac{1}{\binom{K - 2}{l - 1} \cdot (N - 1)^{l}} \\
        ={}& \frac{1}{N(N - 1)^{l}} \cdot  \frac{r^{l} + (-1)^{l} \cdot r}{(r + 1)^{K - 1}}.
    \end{aligned}
\end{equation}
Consequently, for either case, we have
\begin{equation} \label{eq: P_l}
     \bigprob{\que{\dmd,\si}_{n} = q \mid W, S} = \frac{1}{N(N - 1)^{l}} \cdot  \frac{r^{l} + (-1)^{l} \cdot r}{(r + 1)^{K - 1}}.
\end{equation}

Conversely, we consider the case where $n$ is not the inference server, i.e., $\pi(n) \neq 0$.
In this scenario, the Hamming weight of the vectors $\ptrn_{U}$ and $\ptrn_{S}^{(1)}$ is either $(l - 1, 0)$ or $(l - 2, 1)$.
When the pair $(\ptrn_{U}, \ptrn_{S}^{(1)})$ has Hamming weight $(l - 1, 0)$, $\bigprob{\que{\dmd,\si}_{n} = q \mid W, S}$ can be expressed as
\begin{equation}
    \begin{aligned} 
        \bigprob{\que{\dmd,\si}_{n} = q \mid W, S}
        ={}& \frac{1}{N} \cdot P_{l - 1} \cdot \frac{r^{(l - 1) + 1} + (-1)^{(l - 1) + 1} \cdot r}{(r + 1) \cdot r^{(l - 1) + 1}} \cdot \frac{1}{\binom{K - 2}{l - 1} \cdot (N - 1)^{l - 1}} \\
        ={}& \frac{1}{N(N - 1)^{l - 1}} \cdot \frac{r^{l - 1} + (-1)^{l}}{(r + 1)^{K - 1}},
    \end{aligned}    
\end{equation}
and when the pair $(\ptrn_{U}, \ptrn_{S}^{(1)})$ has Hamming weight $(l - 2, 1)$, $\bigprob{\que{\dmd,\si}_{n} = q \mid W, S}$ can be expressed as
\begin{equation}
    \begin{aligned} 
        \bigprob{\que{\dmd,\si}_{n} = q \mid W, S}
        ={}& \frac{1}{N} \cdot P_{l - 2} \cdot \frac{r^{(l - 2) + 1 + 1} + (-1)^{(l - 2) + 1 + 1} \cdot r}{(r + 1) \cdot r^{(l - 2) + 1}} \cdot \frac{1}{\binom{K - 2}{l - 2} \cdot (N - 1)^{l - 1}} \\
        ={}& \frac{1}{N} \cdot P_{l - 2} \cdot \frac{r^{l} + (-1)^{l} \cdot r}{(r + 1) \cdot r^{l - 1}} \cdot \frac{1}{\binom{K - 2}{l - 2} \cdot (N - 1)^{l - 1}} \\
        ={}& \frac{1}{N(N - 1)^{l - 1}} \cdot \frac{r^{l - 1} + (-1)^{l}}{(r + 1)^{K - 1}}.
    \end{aligned}
\end{equation}
Consequently, for either case, we have
\begin{equation} \label{eq: P_l-1}
    \bigprob{\que{\dmd,\si}_{n} = q \mid W, S} = \frac{1}{N(N - 1)^{l - 1}} \cdot \frac{r^{l - 2} + (-1)^{l - 1}}{(r + 1)^{K - 1}}.
\end{equation}

\textbf{Case 3: $w(q) = K$.}

In this case, all indices in the query are non-zero, indicating that the server $n$ must be a non-inference server.
Thus, $\bigprob{\que{\dmd,\si}_{n} = q \mid W, S}$ can be expressed as
\begin{equation}
    \begin{aligned}
        \bigprob{\que{\dmd,\si}_{n} = q \mid W, S}
        ={} &\frac{1}{N} \cdot P_{K - 2} \cdot \frac{r^{(K - 2) + 1 + 1} + (-1)^{(K - 2) + 1 + 1} \cdot r}{(r + 1) \cdot r^{K - 1}} \cdot \frac{1}{(N - 1)^{K - 1}} \\
        ={} &\frac{1}{N} \cdot P_{K - 2} \cdot \frac{r^{K} + (-1)^{K} \cdot r}{(r + 1) \cdot r^{K - 1}} \cdot \frac{1}{(N - 1)^{K - 1}} \\
        ={} &\frac{1}{N(N - 1)^{K - 1}} \cdot \frac{r^{K - 1} + (-1)^{K}}{(r + 1)^{K - 1}}.
    \end{aligned}
\end{equation}

For any two pairs of demand and side information indices $(W, S)$ and $(W', S')$ and any realization $q \in \qspace$ of the query at server $n$, the conditional probability
$\bigprob{\que{\dmd,\si}_n = q \mid W, S}$ can only arise from the exactly one equation in Cases 1 and 3, or from Eq.~\eqref{eq: P_l} and Eq.~\eqref{eq: P_l-1} in Case 2.
Therefore, for any $q \in \qspace$, the ratio
\begin{equation} \label{P_l vs P_l-1}
    \frac{\bigprob{\que{\dmd,\si}_{n} = q \mid W, S}}{\bigprob{\que{\dmd,\si}_{n} = q \mid W', S'}}
\end{equation}
can take only three possible values: either $1$, or
\begin{equation}
    \frac{1}{N - 1} \cdot \frac{r^{l} + (-1)^{l} \cdot r}{r^{l - 1} + (-1)^{l}} = \frac{r}{N - 1},
\end{equation}
or its reciprocal for some $l \in [K - 1]$.
By substituting $r = (N - 1)\ee^{-\varepsilon}$, Eq.~\eqref{P_l vs P_l-1} can be simplified as $\ee^{-\varepsilon}$.
This completes the proof of $\varepsilon$-leaky $(W, S)$-privacy for the scheme with $M = 1$ under the pattern $\key_{2}$.
}

\end{document}